%% file: main.tex
\documentclass[11pt,a4paper]{article}
\usepackage[a4paper,margin=1.04in,includefoot]{geometry}
\usepackage[T1]{fontenc}
\usepackage[utf8]{inputenc}
\usepackage{lmodern}
\input{glyphtounicode}
\usepackage{amsmath,amssymb,amsthm,mathtools,mathrsfs}
\usepackage{xurl}
\usepackage[hidelinks]{hyperref}
\allowdisplaybreaks
\newtheorem{theorem}{Theorem}[section]
\newtheorem{lemma}[theorem]{Lemma}
\newtheorem{proposition}[theorem]{Proposition}
\newtheorem{corollary}[theorem]{Corollary}
\theoremstyle{definition}

\theoremstyle{remark}
\newtheorem{remark}[theorem]{Remark}
\newcommand{\X}{\mathscr X_n}
\newcommand{\C}{\mathcal C}
\newcommand{\F}{\mathbb F}
\newcommand{\one}{\mathbf 1}
\newcommand{\wt}{\operatorname{wt}}
\newcommand{\E}{\mathbb E}
\newcommand{\CT}{\operatorname{CT}}
\newcommand{\Dtwo}{D^{L_2}}
\newcommand{\Adual}{A^\bot}
\newcommand{\wh}{\widehat}
\newcommand{\paperparagraph}[1]{\paragraph{\normalfont\itshape #1}}
\title{Perfect Codes As Exact Minimizers Of Quadratic Discrepancy\\
In \(q\)-ary Hamming Spaces}
\author{Aryeh Lev Zabokritskiy (Yohananov)\\[0.45em]
\small Department of Computer Science, MIGAL--Galilee Research Institute,\\
\small Kiryat Shmona, Israel\\
\small Department of Computer Science, Tel-Hai University of Kiryat\\
\small Shmona and the Galilee, Kiryat Shmona, Israel\\[0.45em]
\small \href{mailto:yuhanalev@telhai.ac.il}{\texttt{yuhanalev@telhai.ac.il}}\\
\small \href{https://orcid.org/0000-0003-3151-6192}{ORCID: 0000-0003-3151-6192}}
\date{}
\hypersetup{
  pdflang={en-US},
  pdftitle={Perfect Codes As Exact Minimizers Of Quadratic Discrepancy In q-ary Hamming Spaces},
  pdfauthor={Aryeh Lev Zabokritskiy (Yohananov)},
  pdfsubject={Discrepancy and perfect codes in Hamming spaces},
  pdfkeywords={Hamming space, perfect codes, quadratic discrepancy, energy minimization}
}
\begin{document}
\maketitle
\begin{abstract}
Total quadratic ball discrepancy measures the deviation of codeword counts
from uniformity over all Hamming balls.  We prove that, whenever a Hamming
space and cardinality admit a perfect code, the discrepancy minimizers among
all subsets of that cardinality are precisely the perfect codes.  This
extends Barg's binary minimizing result and characterizes equality over
arbitrary finite alphabets.  At parameter sets that satisfy the
sphere-packing condition and Lloyd's integral-root condition for correcting
one or two errors, we give an explicit lower bound whose attainment is
equivalent to perfect tiling: the discrepancy that a perfect code with these
parameters would have.  For one error, these include parameters at which the
existence of a perfect code is open.  For alphabets of size at least four,
the discrepancy is an affine function of an energy with a strictly completely
monotonic potential; for these alphabets, known universal energy bounds also
give these lower bounds, and, when a perfect code exists, known universal
optimality results also give the characterization of the minimizers.  Over the ternary alphabet, complete monotonicity fails, and
we prove the ternary bounds directly at all Hamming parameters and at the
Golay parameters, the only nontrivial ternary cases.
\end{abstract}
\smallskip

\input{paper_body.tex}

\bibliographystyle{IEEEtran}
\bibliography{references}
\end{document}

%% file: paper_body.tex
\section{Introduction}

A perfect code imposes uniformity at one radius: every point of Hamming
space belongs to exactly one ball around a codeword.  Total quadratic ball
discrepancy measures a different aspect of uniformity, by comparing codeword
counts in balls of every radius with their averages.  Does
exact tiling at one radius minimize this aggregate discrepancy?  More
precisely, if the length and cardinality admit a perfect code, are the
perfect codes exactly the minimizers among all subsets of that cardinality?

Barg proved the minimizing assertion for binary perfect codes
\cite[Sec.~5 and Thm.~5.4]{Barg2021}.  His work develops Stolarsky's
invariance principle for finite metric spaces and gives distance,
ball-intersection, and Fourier--Krawtchouk formulas in binary Hamming space.
The classical principle relates spherical-cap discrepancy to a defect in
average distance \cite{Stolarsky1973,BrauchartDick2013,BilykDaiMatzke2018}.
Barg proposed working out the details of the invariance principle in
nonbinary Hamming spaces as a further problem
\cite[Sec.~6.3]{Barg2021}; asymptotic bounds for ball discrepancies were
subsequently studied by Barg and Skriganov \cite{BargSkriganov2021}.
Here we answer the fixed-cardinality question over every finite alphabet,
including the characterization of all equality cases.

Let \(G\) be a finite abelian group of order \(q\geq2\), let \(n\geq1\),
and put \(\X\triangleq G^n\).  For a nonempty code \(\C\subseteq\X\), write
\(N\triangleq|\C|\).  A Hamming ball \(B(x,t)\) has volume
\[
 V_t\triangleq |B(x,t)|
 =\sum_{j=0}^{t}\binom nj(q-1)^j.
\]
Its proportion of the ambient space is \(V_t/q^n\), while the proportion of
codewords it contains is \(|\C\cap B(x,t)|/N\).  Summing the squared
differences over all centers and radii defines
\begin{equation}
\label{eq:intro-discrepancy}
 \Dtwo(\C)
 \triangleq\sum_{t=0}^n\sum_{x\in\X}
 \left(
 \frac{|\C\cap B(x,t)|}{N}-\frac{V_t}{q^n}
 \right)^2.
\end{equation}
For \(q=2\), this is Barg's discrepancy \cite[Eqs.~(1)--(2)]{Barg2021}.
The objective depends only on the alphabet size: coordinatewise relabeling
preserves the Hamming metric, balls, and discrepancy.  We fix \(q,n,N\)
and minimize over all subsets of the prescribed cardinality:
\begin{equation}
\label{eq:intro-extremal}
 \Dtwo(q,n,N)
 \triangleq
 \min_{\substack{\C\subseteq G^n\\|\C|=N}}\Dtwo(\C).
\end{equation}

A code \(\mathcal P\subseteq G^n\) is perfect with packing radius
\(e\in\{1,\ldots,n\}\) if its balls of radius \(e\) partition the space.
Equivalently,
\[
 |\mathcal P|V_e=q^n
 \qquad\text{and}\qquad
 B(z,e)\cap B(z',e)=\varnothing\quad(z,z'\in\mathcal P,\ z\neq z').
\]
For a fixed cardinality, all perfect codes have the same packing radius,
because \(V_e\) is strictly increasing in \(e\).  Our main result states that,
whenever a perfect code of cardinality \(N\) exists, a competing code attains
\(\Dtwo(q,n,N)\) if and only if it is perfect.  Thus exact uniformity at one
radius characterizes the minimizers of a quantity that aggregates all radii.

For one or two errors, we prove a quantitative form of the main result that
does not presuppose existence.  We call \((q,n,N)\) \emph{admissible} for
\(e\in\{1,2\}\) (with \(n\geq3\) when \(e=2\), since \(n=2\) forces \(N=1\)) if it satisfies two classical
necessary conditions for a perfect code correcting \(e\) errors: the
sphere-packing condition \(NV_e=q^n\), and Lloyd's integral-root condition that
the Lloyd polynomial of degree \(e\) have \(e\) distinct integral roots in
\(\{1,\ldots,n\}\) \cite{Lloyd1957,Lenstra1972,Delsarte1973}.  At admissible parameters, the
\emph{benchmark} is the discrepancy that a perfect code with these parameters
would have.  It is computed from the parameters alone: discrepancy depends
only on the distance distribution, and the parameters determine this
distribution for every such code.  It bounds the discrepancy of
every \(N\)-point code from below, and a code attains it exactly when its balls
of radius \(e\) tile the space.

\paperparagraph{Energy bounds and the present discrepancy problem}
A related extremal theory studies codes that minimize potential energy.
Ashikhmin and Barg
proved binomial-moment extremality for one-error perfect codes and the Golay
codes \cite{AshikhminBarg1999}.  Cohn and Zhao formulated this property as
universal optimality for completely monotonic potentials, developed
MacWilliams duality for energy-minimizing quasicodes, and gave criteria
covering these families \cite[Lemma~4 and Props.~9, 23, and 24]{CohnZhao2014}.
Boyvalenkov, Dragnev, Hardin, Saff, and Stoyanova obtained universal lower
energy bounds from Levenshtein quadrature
\cite[Thm.~4.2]{BoyvalenkovEtAl2017}.

By Barg's invariance principle \cite[Thm.~2.1]{Barg2021}, total ball
discrepancy is an affine function of the energy of the code under an explicit
potential, the discrepancy potential.  For \(q\geq4\), we prove
that this potential is strictly completely monotonic; for these alphabets,
every universally optimal code therefore minimizes the discrepancy among codes
of its cardinality, and the universal lower energy bounds of
\cite[Thm.~4.2]{BoyvalenkovEtAl2017}, applied to the dual distance
distribution, already give the one- and two-error benchmarks as lower bounds.
The analysis here supplies the exact formulas relating the discrepancy to the
energies of the code and of its dual distance distribution, and the
characterization of equality by perfect codes without
assuming that a perfect code exists.  When a perfect code does exist at the
given parameters, it is universally optimal by Cohn and Zhao's criteria
\cite[Props.~9, 23, and~24]{CohnZhao2014}, and strict complete monotonicity
then forces every minimizer to share its distance distribution
\cite[Lemma~4 and the discussion after Prop.~8]{CohnZhao2014}, hence to be
perfect.

The ternary alphabet requires separate proofs because its discrepancy
potential fails complete monotonicity, as the binary one does
\cite[Sec.~1]{Barg2021};
we prove the bounds at all ternary admissible parameters, namely the Hamming
and Golay parameters.  For \(q=2\), where the admissible parameters are those
of the Hamming codes and of the repetition code of length five, we recover
Barg's minimizing results for perfect codes, characterize all equality cases,
and give a new proof of minimality at the binary Golay parameters, with the
exact minimum value.

\paperparagraph{Parameter range}
For prime-power alphabet sizes, the possible parameters of nontrivial perfect
codes are classical \cite{Tietavainen1973,ZinovievLeontiev1973}:
\[
\begin{array}{lll}
\text{\(q\)-ary Hamming parameters}
 & n=(q^m-1)/(q-1), & e=1,\\[1mm]
\text{ternary Golay parameters}
 & (q,n,N)=(3,11,3^6), & e=2,\\[1mm]
\text{binary Golay parameters}
 & (q,n,N)=(2,23,2^{12}), & e=3.
\end{array}
\]
Together with binary repetition codes of odd length and one-point codes, these
exhaust the possible prime-power parameter families.  The classification
concerns parameters; it allows nonlinear perfect codes with Hamming
parameters.

For non-prime-power alphabets, nontrivial perfect codes have been ruled out
for packing radius at least three \cite{Reuvers1977,Best1983,Hong1984}; the
one- and two-error cases remain open \cite{CazorlaGarcia2024}.  For one error,
the admissible parameters include the lengths \(n=(q^m-1)/(q-1)\), \(m\geq2\),
for every such \(q\), and the few one-error results known over these alphabets
are nonexistence results \cite[Sec.~1.1]{CazorlaGarcia2024}.  Over every alphabet, the known perfect codes with
packing radius two and more than one codeword have the parameters of the
ternary Golay code or of the binary repetition code of length five.  No other
admissible two-error parameter set is known \cite[Eqs.~(5)--(6)]{Bennett2026}.
Bennett \cite[Thms.~2 and~4]{Bennett2026} recently excluded nontrivial
perfect two-error-correcting codes over every non-prime-power alphabet size
\(q=2^{\alpha}p^{\beta}\) with \(p\) an odd prime, \(\alpha,\beta\geq1\), and
\(\alpha\leq20\), and over every non-prime-power \(q\) whose greatest prime
factor is at most \(13\).  The benchmark theorems apply at every admissible
parameter set.  The global characterization follows from the
one- and two-error results, the known classification at larger radii, and a separate treatment of the binary Golay and repetition codes that
this classification leaves.

\paperparagraph{Our contributions}
\begin{itemize}
\item Perfect codes are exactly the discrepancy minimizers whenever the
fixed length and cardinality admit a perfect code, over any finite alphabet.
\item At admissible one- and two-error parameters, an explicit lower bound,
the benchmark, is attained exactly by perfect codes.
\item Exact \(q\)-ary distance and Fourier formulas express total ball
discrepancy as an energy and as a weighted sum of smoothing errors.
\end{itemize}
The last interpretation compares the distribution of a uniform random
codeword, perturbed by uniform noise on a Hamming ball, with the uniform
distribution on the ambient space.  Each radius \(t\) contributes the chi-squared divergence of
this smoothed distribution from the uniform distribution, and the discrepancy
is the sum of these divergences with weights \(V_t^2/q^n\).
This smoothing connection was studied for binary codes in
\cite{PathegamaBarg2023}; here it is stated for \(q\)-ary codes in the
normalization of \eqref{eq:intro-discrepancy}.

\paperparagraph{Organization}
Section~\ref{sec:prelim} states the global characterization and introduces
the notation and the one- and two-error benchmarks.
Section~\ref{sec:qary-reduction} derives the distance and Fourier formulas
for the discrepancy and analyzes the weights in the Fourier formula.
Section~\ref{sec:hamming} proves the one-error bound, and
Section~\ref{sec:golay} treats multiple errors and completes the global
argument.  Section~\ref{sec:interpretation} gives the smoothing
interpretation.

\section{Main results and notation}
\label{sec:prelim}

The main result characterizes all minimizers whenever the prescribed
cardinality admits a perfect code.  Its statement uses only the definitions
of the discrepancy and of perfect codes from the Introduction.

\begin{theorem}[Perfect codes are exactly the discrepancy minimizers]
\label{thm:all-perfect}
Let \(G\) be any finite abelian group, and let
\(\mathcal P\subseteq G^n\) be a perfect code.  Then every code
\(\C\subseteq G^n\) with \(|\C|=|\mathcal P|\) satisfies
\[
 \Dtwo(\C)\geq\Dtwo(\mathcal P).
\]
Equality holds if and only if \(\C\) is perfect.
\end{theorem}

No additivity or linearity is assumed of either code.  Any perfect code of
the same cardinality as \(\mathcal P\) necessarily has the same packing
radius, because \(|\mathcal P|V_e=q^n\) and the ball volumes \(V_e\)
are strictly increasing.  Theorems~\ref{thm:all-one-perfect} and
\ref{thm:two-perfect-all-q} prove the global theorem for packing radius one
and two.  Theorem~\ref{thm:binary-golay} for the binary Golay parameters and
Proposition~\ref{prop:repetition} for the repetition codes, together with the
known parameter classification and the nonexistence results for
non-prime-power alphabets at packing radius at least three, complete the
remaining cases; the case \(|\mathcal P|=1\) is immediate.

The next subsection introduces the Fourier notation needed to state the one-
and two-error bounds; Sections~\ref{sec:qary-reduction}--\ref{sec:golay}
prove them.

\subsection{Hamming and Fourier notation}

Let \(G\) be a finite abelian group of order \(q\geq2\), let \(n\geq1\), and put
\(\X\triangleq G^n\), so that \(|\X|=q^n\).  Any \(q\)-symbol alphabet may be identified
with such a group by relabeling its symbols; this preserves Hamming distance
and discrepancy.  The group structure is used only for Fourier analysis.  The
notation \([n]\) denotes \(\{1,\ldots,n\}\).  The
Hamming metric
\[
 d(x,y)\triangleq|\{i:x_i\neq y_i\}|,\qquad
 \wt(x)\triangleq d(x,0)=|\{i:x_i\neq0\}|,
\]
depends only on the alphabet.  Write
\[
 B(x,t)\triangleq\{y\in\X:d(x,y)\leq t\},\qquad
 V_t\triangleq|B(x,t)|=\sum_{j=0}^t\binom nj(q-1)^j.
\]
A code is an arbitrary nonempty subset \(\C\subseteq\X\), not necessarily a
subgroup, and \(N\triangleq|\C|\).
For \(1\leq e\leq n\), a code is perfect \(e\)-error-correcting if the balls
\(B(z,e)\), \(z\in\C\), partition \(G^n\).  The discrepancy sums over every
radius \(t\); the packing radius \(e\) of a perfect code is determined by its
cardinality through \(NV_e=q^n\).

For a subset \(S\subseteq\X\), its indicator is the function
\(\one_S:\X\to\{0,1\}\) defined by
\[
 \one_S(x)\triangleq
 \begin{cases}
  1,&x\in S,\\
  0,&x\notin S.
 \end{cases}
\]
We use the same notation \(\one_{\{P\}}\) for the indicator of a condition
\(P\).

The distance distribution of \(\C\) is
\[
 A_i\triangleq\frac1N|\{(z,z')\in\C^2:d(z,z')=i\}|,\qquad
 0\leq i\leq n.
\]
Thus \(A_0=1\) and \(\sum_iA_i=N\).  It determines
\(\sum_{z,z'\in\C}h(d(z,z'))\) for every function \(h\) on
\(\{0,\ldots,n\}\).  This is the normalization used in
\cite{Barg2021}.

Let
\[
 \widehat G\triangleq\{\gamma:G\longrightarrow\mathbb C:
   \gamma(a+b)=\gamma(a)\gamma(b),\ |\gamma(a)|=1\}
\]
be the character group of \(G\).  A frequency
\(\xi=(\xi_1,\ldots,\xi_n)\in\widehat G^{\,n}\) determines the character
\[
 \chi_\xi(x)\triangleq\prod_{i=1}^n\xi_i(x_i),\qquad x\in\X.
\]
Thus \(x\) is a point of the Hamming space, whereas \(\xi\) labels a Fourier
frequency.  Let \(\gamma_0\in\widehat G\) denote the trivial character,
defined by \(\gamma_0(a)=1\) for every \(a\in G\).  We call a coordinate
\(\xi_i\) nontrivial if \(\xi_i\neq\gamma_0\), and set
\[
 \wt(\xi)\triangleq|\{i:\xi_i\neq\gamma_0\}|.
\]

Let \(f:\X\to\mathbb C\) be an arbitrary complex-valued function.
Its Fourier transform \(\wh f\) is the function on
\(\widehat G^{\,n}\) defined by
\[
 \wh f(\xi)\triangleq
 \sum_{x\in\X}f(x)\overline{\chi_\xi(x)}.
\]
For functions \(f,g:\X\to\mathbb C\), define their convolution by
\[
 (f*g)(x)\triangleq\sum_{y\in\X}f(y)g(x-y),
\]
for which
\[
 \wh{f*g}=\wh f\,\wh g,\qquad
 \sum_x|f(x)|^2=\frac1{q^n}\sum_\xi|\wh f(\xi)|^2.
\]
These Fourier conventions and the Parseval identity in the last display are standard for
finite abelian groups; see \cite[Part~I]{Terras1999}.

The \(q\)-ary Krawtchouk polynomials are
\begin{equation}
\label{eq:kraw}
 K_j^{(n,q)}(w)
 \triangleq\sum_{\ell=0}^j(-1)^\ell(q-1)^{j-\ell}
   \binom w\ell\binom{n-w}{j-\ell}.
\end{equation}
They are characterized by
\begin{equation}
\label{eq:kraw-gen}
 \sum_{j=0}^nK_j^{(n,q)}(w)z^j
   =(1+(q-1)z)^{n-w}(1-z)^w.
\end{equation}
The Fourier transform of the indicator of the weight-\(j\) sphere, evaluated
at a character of weight \(w\), is \(K_j^{(n,q)}(w)\).  These are standard
facts about the \(q\)-ary Hamming scheme; see
\cite{Delsarte1973,MacWilliamsSloane1977}.  For \(q=2\), they are the
Krawtchouk conventions of \cite[Sec.~3.1]{Barg2021}.

For a Laurent polynomial \(P(z)\), the notation \([z^j]P(z)\) denotes the
coefficient of \(z^j\).  Its constant term is denoted by
\[
 \CT_z P(z)\triangleq[z^0]P(z).
\]

Define the dual distance distribution of \(\C\) by
\begin{equation}
\label{eq:dual}
 \Adual_w(\C)\triangleq\frac1{N^2}
   \sum_{\substack{\xi\in\widehat G^{\,n}\\\wt(\xi)=w}}
         |\wh{\one_\C}(\xi)|^2,\qquad 0\leq w\leq n.
\end{equation}
Thus \(\Adual_w\geq0\), \(\Adual_0=1\), and, by Parseval, the total mass is
\(\sum_{w=0}^n\Adual_w=q^n/N\); Theorem~\ref{thm:spectral} will express
\(\Dtwo(\C)\) as a weighted sum of these nonnegative masses.  Character
orthogonality also gives the usual
MacWilliams--Delsarte form
\begin{equation}
\label{eq:macwilliams}
 \Adual_w=\frac1N\sum_{i=0}^nA_iK_w^{(n,q)}(i).
\end{equation}
Equations \eqref{eq:dual}--\eqref{eq:macwilliams} are the \(q\)-ary extension
of the binary dual distribution in
\cite[Sec.~3.2]{Barg2021}.  For a nonlinear code, the superscript denotes the
MacWilliams--Delsarte transform and does not assert the existence of a dual
code.  Following \cite[Sec.~III]{CohnZhao2014}, we call
\(\{w\geq1:\Adual_w(\C)>0\}\) the \emph{dual support} of \(\C\).

For \(0\leq w,t\leq n\), define the Fourier coefficient of a radius-\(t\)
ball at character weight \(w\) by
\begin{equation}
\label{eq:c-def}
 c_w(t)\triangleq\sum_{j=0}^tK_j^{(n,q)}(w).
\end{equation}
For \(1\leq w\leq n\), define the discrepancy spectral weight by
\begin{equation}
\label{eq:W-def}
 W_{n,q}(w)\triangleq\sum_{t=0}^{n-1}c_w(t)^2.
\end{equation}
Section~\ref{sec:qary-reduction} will show why this weight governs
discrepancy and will derive the form used to analyze its shape.

\subsection{One- and two-error benchmarks}

Recall that \(\Dtwo(\C)\) is the total quadratic ball discrepancy defined in
\eqref{eq:intro-discrepancy}.

For the one-error benchmark, let \(q\geq2\) and \(n\geq1\), and put
\[
 V\triangleq1+n(q-1),\qquad
 N\triangleq\frac{q^n}{V},\qquad
 k\triangleq\frac Vq,
\]
so that \(V=|B(0,1)|\), \(N\) is the cardinality of a perfect
single-error-correcting code, and \(k\) is the root of the one-error Lloyd
polynomial \(c_w(1)=V-qw\).

\begin{theorem}[One-error benchmark]
\label{thm:all-one-perfect}
Let \(q\geq2\) and \(n\geq1\), and assume that \(N\) and \(k\) defined above
are integers.  Then every subset \(\C\subseteq G^n\) of cardinality \(N\)
satisfies
\begin{equation}
\label{eq:main-bound}
 \Dtwo(\C)\geq\frac{V-1}{q^n}W_{n,q}(k),
\end{equation}
with equality if and only if \(\C\) is a perfect single-error-correcting
code.  Thus the benchmark in \eqref{eq:main-bound} is attainable if and
only if such a perfect code exists.
\end{theorem}

For \(q\geq4\), the bound \eqref{eq:main-bound} also follows, after
MacWilliams duality, from the universal lower energy bound
\cite[Thm.~4.2]{BoyvalenkovEtAl2017} at the strength-two Rao endpoint.  Here
the dual distribution has total mass \(V=R(n,2)\), the Rao bound for designs
of strength two \cite[Eq.~(18)]{BoyvalenkovEtAl2017}, at which the
Levenshtein bounds of degrees one and two coincide
\cite[Eq.~(21)]{BoyvalenkovEtAl2017}.  Ashikhmin and Barg
\cite{AshikhminBarg1999} and Cohn and Zhao \cite[Props.~9 and 24]{CohnZhao2014}
give the corresponding binomial-moment extremality and universal optimality of
the \(q\)-ary Hamming codes.
Because the discrepancy potential is strictly completely monotonic for
\(q\geq4\) (Proposition~\ref{prop:complete-monotone}), these results also
imply that, at the \(q\)-ary Hamming parameters, every \(N\)-point code
attaining this bound has the distance distribution of a Hamming code and is
therefore perfect \cite[Lemma~4 and the discussion after Prop.~8]{CohnZhao2014}.
Theorem~\ref{thm:all-one-perfect} gives this bound in the discrepancy
normalization and proves the characterization of equality directly for every
\(q\), without assuming that a perfect code exists.  For \(q=2\), Barg proves the numerical
inequality and minimizing conclusion \cite{Barg2021}; the characterization of
equality by perfect codes is proved here.

For the two-error benchmark, let \(q\geq2\) and \(n\geq3\), and put
\[
 V^{(2)}\triangleq V_2=
 1+n(q-1)+\binom n2(q-1)^2,
 \qquad
 N^{(2)}\triangleq\frac{q^n}{V^{(2)}}.
\]
Let
\[
 L_2(w)\triangleq c_w(2)=K_2^{(n-1,q)}(w-1)
\]
be the two-error Lloyd polynomial; the second equality is
Lemma~\ref{lem:ball-transform}.  This notation is unrelated to the
superscript \(L_2\) in \(\Dtwo\), which denotes quadratic discrepancy.
Suppose that \(N^{(2)}\) is an integer
and that \(L_2\) has two distinct integral roots
\(r<s\) in \(\{1,\ldots,n\}\).  Define
\[
 \bar w\triangleq\frac{(q-1)n}{q},
\]
\begin{equation}
\label{eq:formal-two-masses}
 b_r\triangleq
 \frac{(V^{(2)}-1)s-V^{(2)}\bar w}{s-r},
 \qquad
 b_s\triangleq
 \frac{V^{(2)}\bar w-(V^{(2)}-1)r}{s-r},
\end{equation}
equivalently, \(b_r+b_s=V^{(2)}-1\) and \(rb_r+sb_s=V^{(2)}\bar w\); and
\begin{equation}
\label{eq:radius-two-benchmark}
 \delta_{n,q}^{(2)}
 \triangleq
 \frac{b_rW_{n,q}(r)+b_sW_{n,q}(s)}{q^n}.
\end{equation}

\begin{theorem}[Two-error benchmark]
\label{thm:two-perfect-all-q}
Let \(q\geq2\) and \(n\geq3\).  Assume that \(N^{(2)}\) is an integer and
that \(L_2\) has two distinct
integral roots \(r<s\) in \(\{1,\ldots,n\}\).  Every \(N^{(2)}\)-point code
\(\C\subseteq G^n\) satisfies
\begin{equation}
\label{eq:radius-two-bound}
 \Dtwo(\C)\geq\delta_{n,q}^{(2)}.
\end{equation}
Equality holds if and only if \(\C\) is a perfect two-error-correcting code.
Consequently, such a perfect code exists at these parameters if and only if
\[
 \Dtwo(q,n,N^{(2)})=\delta_{n,q}^{(2)}.
\]
\end{theorem}

For \(q\geq4\), the right-hand side of \eqref{eq:radius-two-bound} is the
value of the linear-programming certificate underlying
\cite[Thm.~4.2]{BoyvalenkovEtAl2017} at the strength-four Rao endpoint, where
the dual distribution has total mass \(V^{(2)}=R(n,4)\); the roots \(r,s\) and
masses \(b_r,b_s\) are the nodes and weights of its quadrature in our
normalization.  The proof in Section~\ref{sec:golay} compares an arbitrary code with a vector computed
from \(q\) and \(n\) whose entries may be negative, so it does not assume that
a perfect code exists.  For the two
smaller alphabets, the hypotheses force the parameters \((q,n)=(2,5)\) of the
binary repetition code and \((3,11)\) of the ternary Golay code
(Lemma~\ref{lem:small-q-two-parameters}); the binary inequality is Barg's
repetition case, with its equality case proved here, while the ternary Golay
discrepancy certificate is new.

Both hypotheses in the two-error theorem are necessary for a perfect
two-error-correcting code.  The first is the sphere-packing condition.
Lloyd's theorem states that the Lloyd polynomial of a perfect
\(e\)-error-correcting code has \(e\) distinct integral roots in
\(\{1,\ldots,n\}\); for \(e=2\), this is precisely the integral-root
hypothesis above \cite{Lloyd1957,Lenstra1972,Delsarte1973}.  By
Lemma~\ref{lem:two-perfect-spectrum}, \(b_r\) and \(b_s\) are the values at
the Lloyd roots of the dual distribution of every perfect code with these
parameters.  Apart from the parameters of the binary
repetition code of length five and the ternary Golay code, no parameter set
satisfying both hypotheses is currently known
\cite[Eqs.~(5)--(6)]{Bennett2026}; the theorem applies to any further set
that satisfies them.

We call a one-error parameter set \emph{admissible} when the two quantities
\(N\) and \(k\) above are integers, and a two-error parameter set
\emph{admissible} when \(n\geq3\) and it satisfies the sphere-packing and
integral-root conditions of Theorem~\ref{thm:two-perfect-all-q}.  This
terminology records only necessary arithmetic conditions and does not assert
that a perfect code exists.

\begin{corollary}[The \(q\)-ary Hamming family]
\label{cor:qary-hamming-family}
Let \(q\) be a prime power and \(m\geq2\), and put
\[
 n=\frac{q^m-1}{q-1},\qquad N=q^{n-m}.
\]
Then every \(N\)-point code \(\C\subseteq\F_q^n\) satisfies
\[
 \Dtwo(\C)\geq
 \frac{q^m-1}{q^n}W_{n,q}\bigl(q^{m-1}\bigr).
\]
Equality holds if and only if \(\C\) is a perfect
single-error-correcting code.  Consequently, the minimizers are precisely all
perfect codes with the \(q\)-ary Hamming parameters; linearity is neither
assumed nor forced.
\end{corollary}

\begin{proof}
Here \(1+n(q-1)=q^m\), so the parameters in
Theorem~\ref{thm:all-one-perfect} are
\[
 V=q^m,\qquad k=q^{m-1},\qquad N=q^{n-m}.
\]
The standard \(q\)-ary Hamming code shows that equality is attainable.
\end{proof}

When \(q=2\), the integrality assumptions force \(n=2^m-1\) for some
\(m\geq1\), and Theorem~\ref{thm:all-one-perfect} reads
\begin{equation}
\label{eq:binary-main-bound}
 \Dtwo(\C)\geq
 \frac{n}{2^n}\binom{n-1}{(n-1)/2},
\end{equation}
with equality if and only if \(\C\) is a binary perfect
single-error-correcting code.  The value in \eqref{eq:binary-main-bound} is
evaluated in \cite[Thm.~3.7 and Eq.~(44)]{Barg2021}; its minimizing property
is \cite[Cor.~5.3, Eq.~(70), and Example~2a]{Barg2021}, and the minimizing
statement for all binary perfect codes is
\cite[Thm.~5.4]{Barg2021}.  Barg does not assert the converse
characterization of all equality cases.  Thus \(q=2\) recovers the earlier
binary minimizing result, and for \(q\geq4\) the bound
\eqref{eq:main-bound} also follows from
\cite[Thm.~4.2]{BoyvalenkovEtAl2017} (Section~\ref{sec:universal-optimality});
Theorem~\ref{thm:all-one-perfect} adds the ternary case and, for every \(q\),
the characterization of equality by perfect codes without assuming that a
perfect code exists.  The exact specialization appears in
Appendix~\ref{app:binary-formulas}; Lemma~\ref{lem:binary-sign} gives the
direct proof of the binary slope inequality used~here.

The Fourier formula of Theorem~\ref{thm:spectral} expresses \(\Dtwo(\C)\) as a nonnegative weighted sum of
the dual distance distribution.  For \(n\geq2\), moment constraints and
strict convexity reduce the one-error theorem to the single inequality
\begin{equation}
\label{eq:intro-bottleneck}
 W_{n,q}(k)<W_{n,q}(k-1),
\end{equation}
and only the proof of this inequality depends on the alphabet: monotonicity
for \(q\geq4\), a recurrence with positive coefficients for \(q=3\), and
reflection symmetry for \(q=2\).  For \(q\geq4\), the two-error proof compares
an arbitrary code with a formal distance-distribution vector whose entries may
be negative; equality forces the code to have this distance distribution,
hence minimum distance at least five, so its radius-two balls are disjoint
and, by the sphere-packing condition, tile the space.  In the one-error and
Golay cases, equality forces the dual support into the set of Lloyd roots, and
Proposition~\ref{prop:lloyd-support} turns this into a tiling.  The supporting
lines and polynomial minorants in these arguments are explicit Delsarte
certificates.  The slope inequality at the Lloyd root settles the ternary
Hamming parameters, a quadratic minorant
settles the ternary Golay parameters, and a quartic minorant gives the exact
binary Golay value.

\section{Invariance and spectral reduction}
\label{sec:qary-reduction}

This section derives the two energy representations used below.  The Fourier
form is used for the single-error and Golay certificates, and the distance
form for the large-alphabet two-error theorem and the repetition codes.

\subsection{Metric invariance}
\label{sec:invariance}

Both energy forms come from Stolarsky's invariance principle, which Barg
proved for distance-invariant finite metric spaces; he then computed its
distance kernel in the binary Hamming space
\cite[Thm.~2.1 and Lemma~3.1]{Barg2021}.  Here we specialize the identity to
\(\X=G^n\) and compute the distance kernel for every \(q\).

The identity for the total discrepancy in
\eqref{eq:intro-discrepancy} uses two kernels.  The first compares the
distances from \(x\) and from \(y\) to every point \(u\); the second adds the intersections of
equal-radius balls.  Namely, for \(x,y\in\X\), put
\begin{align}
 \lambda(x,y)
   &\triangleq\frac12\sum_{u\in\X}|d(x,u)-d(y,u)|,\label{eq:lambda-def}\\
 \mu_t(x,y)
   &\triangleq|B(x,t)\cap B(y,t)|,\qquad
 \mu(x,y)\triangleq\sum_{t=0}^n\mu_t(x,y).\label{eq:mu-def}
\end{align}
Both kernels depend only on \(w=d(x,y)\).  Write
\(\lambda(w)\), \(\mu_t(w)\), and \(\mu(w)\) for their radial values, with
the dependence on \(n\) and \(q\) suppressed.  The kernel \(\lambda\) gives
the distance form of the identity, and the Fourier expansion of the
ball-intersection kernel \(\mu\) will produce the spectral weight.
The two kernels satisfy \(\lambda(w)+\mu(w)=\mu(0)\).  Indeed, fix
\(x,y,u\), and put \(a\triangleq d(x,u)\), \(b\triangleq d(y,u)\).  Then
\[
 \sum_{t=0}^n\left(
 \one_{\{a\leq t\}}+\one_{\{b\leq t\}}
 -2\one_{\{a\leq t\}}\one_{\{b\leq t\}}\right)=|a-b|.
\]
Summing over \(u\) gives
\(\mu(x,x)+\mu(y,y)-2\mu(x,y)=2\lambda(x,y)\), and
\(\mu(x,x)=\mu(y,y)=\mu(0)\).
For a nonempty set \(S\subseteq\X\) and a kernel \(h\) on \(\X\), write
\[
 \langle h\rangle_S\triangleq
 \frac1{|S|^2}\sum_{x,y\in S}h(x,y)
\]
for its average over ordered pairs from \(S\).

\begin{samepage}
\begin{theorem}[Stolarsky's invariance principle in the \(q\)-ary Hamming space]
\label{thm:invariance}
Let \(\C\subseteq G^n\) have size \(N\).  Then
\begin{equation}
\label{eq:invariance}
 \Dtwo(\C)
 =\langle\lambda\rangle_{\X}-\langle\lambda\rangle_{\C}
 =\langle\mu\rangle_{\C}-\langle\mu\rangle_{\X},
\end{equation}
where the subscripts denote averages over ordered pairs.  In particular,
\begin{equation}
\label{eq:invariance-distance}
 \Dtwo(\C)=\Lambda_{n,q}
       -\frac1N\sum_{w=1}^n A_w\lambda(w),
\end{equation}
where
\begin{equation}
\label{eq:Lambda}
 \Lambda_{n,q}
   \triangleq\frac1{q^n}\sum_{w=0}^n\binom nw(q-1)^w\lambda(w).
\end{equation}
\end{theorem}
\end{samepage}

\begin{proof}
For fixed \(t\), put \(a_x\triangleq|\C\cap B(x,t)|\).  Then
\(\sum_xa_x=NV_t\) and \(\sum_xa_x^2=\sum_{z,z'\in\C}\mu_t(z,z')\), while
\(\langle\mu_t\rangle_{\X}=V_t^2/q^n\) because every point lies in exactly
\(V_t\) balls of radius \(t\).  Expanding the square in
\eqref{eq:intro-discrepancy} therefore gives
\[
 \sum_{x\in\X}\Bigl(\frac{a_x}N-\frac{V_t}{q^n}\Bigr)^2
 =\langle\mu_t\rangle_{\C}-\langle\mu_t\rangle_{\X},
\]
and summing over \(t\) gives the second form in \eqref{eq:invariance}.  The
relation \(\lambda(w)+\mu(w)=\mu(0)\) turns it into the first form.  Both forms
are the specialization to \(\X\) of the finite-metric identity and its
ball-intersection form in \cite[Sec.~2, especially Thm.~2.1]{Barg2021}.  Since
\(\lambda(x,y)=\lambda(d(x,y))\), grouping ordered pairs of codewords by
their distance gives \eqref{eq:invariance-distance}.  Finally, a fixed point
has \(\binom nw(q-1)^w\) points at distance \(w\), which gives
\eqref{eq:Lambda}.
\end{proof}

\subsection{The distance kernel}

The increment law below is used in Propositions~\ref{prop:complete-monotone}
and~\ref{prop:repetition}.  A full multinomial formula
and its binary specialization are recorded in
Appendix~\ref{app:binary-formulas}.

For \(j\geq0\), define
\begin{equation}
\label{eq:Tj}
 T_j(q)\triangleq\CT_z(z+z^{-1}+q-2)^j
 =\sum_{a=0}^{\lfloor j/2\rfloor}
   \frac{j!}{a!^2(j-2a)!}(q-2)^{j-2a}.
\end{equation}

\begin{proposition}[Distance-kernel increments]
\label{prop:lambda}
For \(0\leq w\leq n-1\),
\begin{equation}
\label{eq:lambda-increment}
 \lambda(w+1)-\lambda(w)
    =q^{n-w-1}T_w(q).
\end{equation}
Consequently, for \(0\leq w\leq n\),
\begin{equation}
\label{eq:lambda-sum}
 \lambda(w)
   =q^{n-1}\sum_{j=0}^{w-1}q^{-j}T_j(q).
\end{equation}
In particular, \(\lambda\) is strictly increasing when \(q>2\).  The binary
plateaus are described in Appendix~\ref{app:binary-formulas}.
\end{proposition}

\begin{proof}
Let \(X_1,X_2,\ldots\) be independent with probabilities \(1/q,1/q\), and
\((q-2)/q\) at \(1,-1\), and \(0\), respectively, and put
\(S_w\triangleq X_1+\cdots+X_w\).  Comparing two words at distance \(w\) coordinate by
coordinate gives \(\lambda(w)=q^n\E|S_w|/2\).  For every integer \(s\),
\[
 \E\bigl(|s+X_{w+1}|-|s|\bigr)
 =\begin{cases}2/q,&s=0,\\0,&s\neq0.\end{cases}
\]
Consequently,
\[
 \lambda(w+1)-\lambda(w)
 =q^{n-1}\Pr(S_w=0).
\]
Since \(\Pr(S_w=0)=q^{-w}T_w(q)\), this proves
\eqref{eq:lambda-increment}; summing the increments proves
\eqref{eq:lambda-sum}.  Since \(T_w(q)>0\) for \(q>2\), the strict
increase follows.
\end{proof}

\subsection{The spectral energy}
\label{sec:fourier}

The Fourier representation is used for the extremal results.  After summing
over all radii, the squared transforms of the balls give one number
\(W_{n,q}(w)\) for each character weight \(w\), and the geometric
optimization problem becomes a moment problem for the nonnegative dual
distribution.

Recall from \eqref{eq:c-def} that \(c_w(t)\) is the Fourier coefficient of a
radius-\(t\) ball at character weight \(w\).  For \(q=2\), it is the
coefficient denoted \(c_k(t)\) in \cite[Sec.~4]{Barg2021}.  As a function of
\(w\), \(c_w(t)\) is the Lloyd polynomial of degree \(t\), and the generating
function \eqref{eq:kraw-gen} expresses it as a Krawtchouk polynomial of
length \(n-1\).  This identity is classical; see \cite{Delsarte1973} and, in
the binary normalization, \cite[Lemma~4.1, Eqs.~(50)--(52)]{Barg2021}.

\begin{lemma}[Fourier transform of a Hamming ball]
\label{lem:ball-transform}
For \(1\leq w\leq n\) and \(0\leq t\leq n-1\),
\begin{equation}
\label{eq:lloyd-dimension-drop}
 c_w(t)=K_t^{(n-1,q)}(w-1).
\end{equation}
Moreover, \(c_w(n)=0\).
\end{lemma}

\begin{proof}
Multiplying \eqref{eq:kraw-gen} by \((1-z)^{-1}\), we find
\[
 \sum_{t\geq0}c_w(t)z^t
 =(1+(q-1)z)^{n-w}(1-z)^{w-1}.
\]
The coefficient of \(z^t\) is the right-hand side of
\eqref{eq:lloyd-dimension-drop}.  The last assertion is the vanishing of a
nontrivial character summed over the entire space.
\end{proof}

Combining \eqref{eq:W-def} with Lemma~\ref{lem:ball-transform} gives
\begin{equation}
\label{eq:W-kraw-squares}
 W_{n,q}(w)=\sum_{t=0}^{n-1}c_w(t)^2
 =\sum_{t=0}^{n-1}\bigl(K_t^{(n-1,q)}(w-1)\bigr)^2.
\end{equation}

The next subsection analyzes the shape of the nonnegative sequence
\(W_{n,q}\).  Its relation to the Krawtchouk coefficients of \(\lambda\) used
in \cite{Barg2021} is as follows.  If
\[
 \lambda(w)=\sum_{r=0}^n\wh\lambda_r^{(n,q)}
 K_r^{(n,q)}(w)
\]
is the corresponding Krawtchouk expansion, then
\begin{equation}
\label{eq:W-lambda-coefficient}
 \wh\lambda_0^{(n,q)}=\Lambda_{n,q},\qquad
 \wh\lambda_r^{(n,q)}=-q^{-n}W_{n,q}(r)\quad(1\leq r\leq n).
\end{equation}
Indeed, Fourier inversion of the ball-intersection kernel gives
\[
 \mu_t(w)=q^{-n}\sum_{r=0}^n c_r(t)^2K_r^{(n,q)}(w).
\]
After summing over \(t\), the nonconstant Krawtchouk coefficient of \(\mu\)
at weight \(r\) is \(q^{-n}W_{n,q}(r)\), because \(c_r(n)=0\) for
\(r\geq1\).  Since \(\lambda(w)+\mu(w)=\mu(0)\) (Section~\ref{sec:invariance}),
the nonconstant coefficients of
\(\lambda\) are the negatives of those of \(\mu\).  The constant coefficient
of \(\lambda\) is its average \(\Lambda_{n,q}\) over the Hamming space,
proving
\eqref{eq:W-lambda-coefficient}.

The next theorem extends Barg's binary spectral formula
\cite[Thm.~4.4, Eq.~(57)]{Barg2021} to the present \(q\)-ary normalization.

\begin{theorem}[Spectral discrepancy formula]
\label{thm:spectral}
For every nonempty code \(\C\subseteq G^n\),
\begin{equation}
\label{eq:spectral-discrepancy}
 \Dtwo(\C)=\frac1{q^n}\sum_{w=1}^n\Adual_w(\C)W_{n,q}(w).
\end{equation}
\end{theorem}

\begin{proof}
For fixed \(t\), the expression inside the square in
\eqref{eq:intro-discrepancy} is
\[
 g_t\triangleq
 \frac1N\one_\C*\one_{B(0,t)}-\frac{V_t}{q^n}\one_{\X}.
\]
Its Fourier transform vanishes at the trivial character and, at every
nontrivial frequency \(\xi\), equals
\[
 \wh g_t(\xi)=\frac1N\wh{\one_\C}(\xi)c_{\wt(\xi)}(t).
\]
Parseval, Lemma~\ref{lem:ball-transform}, and summation over \(t\) give
\eqref{eq:spectral-discrepancy}.
\end{proof}

\subsection{The shape of the spectral weight}

The problem is now reduced to minimizing \(\sum_w\Adual_wW_{n,q}(w)\) over
nonnegative vectors \(\Adual\) with constrained moments.  The next theorem and
its corollary give the convexity and monotonicity of \(W_{n,q}\) used in this
minimization.

In the binary case, Barg \cite{Barg2021} derived a closed form for
\eqref{eq:W-def}.  Identities for sums of squared Krawtchouk polynomials
have a substantial literature; see, for instance,
\cite{Dette1994}.  We use instead the following integral representation,
which retains the dependence on \(q\).

For \(0\leq\theta\leq2\pi\), define
\[
 \alpha_q(\theta)\triangleq|1+(q-1)e^{i\theta}|^2,\qquad
 \beta(\theta)\triangleq|1-e^{i\theta}|^2.
\]

\begin{theorem}[Integral representation]
\label{thm:integral-representation}
For \(1\leq w\leq n\),
\begin{equation}
\label{eq:W-integral}
 W_{n,q}(w)=\frac1{2\pi}\int_0^{2\pi}
 \alpha_q(\theta)^{\,n-w}\beta(\theta)^{\,w-1}\,d\theta.
\end{equation}
\end{theorem}

\begin{proof}
Put \(m\triangleq n-1\) and \(x\triangleq w-1\).  By \eqref{eq:kraw-gen}, the values
\(K_t^{(m,q)}(x)\) are the coefficients of
\[
 P_x(z)\triangleq(1+(q-1)z)^{m-x}(1-z)^x.
\]
Parseval on the unit circle gives
\[
 \sum_{t=0}^m|K_t^{(m,q)}(x)|^2
 =\frac1{2\pi}\int_0^{2\pi}|P_x(e^{i\theta})|^2\,d\theta,
\]
which is \eqref{eq:W-integral} by \eqref{eq:W-kraw-squares}.
\end{proof}

\begin{corollary}[Spectral shape]
\label{cor:convexity}
The sequence \(W_{n,q}(1),\ldots,W_{n,q}(n)\) is discretely convex and is
strictly convex whenever \(n\geq3\).  More precisely, for
\(1\leq w\leq n-2\),
\begin{equation}
\label{eq:W-convex}
\begin{aligned}
 &W_{n,q}(w+2)-2W_{n,q}(w+1)+W_{n,q}(w)\\
 &\qquad=\frac1{2\pi}\int_0^{2\pi}
 \alpha_q^{\,n-w-2}\beta^{\,w-1}
  (\beta-\alpha_q)^2\,d\theta>0.
\end{aligned}
\end{equation}
For \(q\geq4\), the sequence is strictly decreasing.
\end{corollary}

\begin{proof}
Taking the second finite difference in \eqref{eq:W-integral} yields
\eqref{eq:W-convex}.  Away from finitely many values of \(\theta\), all three
factors in its integrand are positive.

For the first difference,
\begin{equation}
\label{eq:W-first-difference}
 W_{n,q}(w+1)-W_{n,q}(w)
 =\frac1{2\pi}\int_0^{2\pi}
 \alpha_q^{\,n-w-1}\beta^{\,w-1}
  (\beta-\alpha_q)\,d\theta.
\end{equation}
Writing \(y=\sin^2(\theta/2)\), we have
\[
 \beta-\alpha_q=4y-\bigl(q^2-4(q-1)y\bigr)=q(4y-q).
\]
For \(q\geq4\), this is nonpositive for all \(\theta\) and is negative on a
set of positive measure.  Equation \eqref{eq:W-first-difference} is therefore
strictly negative.
\end{proof}

\section{Perfect codes correcting one error}
\label{sec:hamming}

This section reduces Theorem~\ref{thm:all-one-perfect} to the slope
inequality \eqref{eq:intro-bottleneck}.  For alphabets of size at least four,
the inequality follows immediately from Corollary~\ref{cor:convexity};
Subsection~\ref{sec:small-alphabets} proves it for \(q=2,3\).
Lemma~\ref{lem:moments} also records the second dual moment, which
Section~\ref{sec:golay} uses for the ternary Golay parameters.

The case \(n=1\) is immediate: then \(N=1\), and every one-point code is
perfect because its ball of radius one is the whole space.  We therefore
assume \(n\geq2\) in the rest of this section.

\subsection{Moment and support constraints}

Recall that \(\wt(x)\) is the Hamming weight of a point \(x\in G^n\), while
\(\wt(\xi)\) is the character weight of a Fourier frequency
\(\xi\in\widehat G^{\,n}\), namely, the number of its nontrivial coordinate
characters.

The next identities are the first Pless--Delsarte power-moment identities in the
present dual-distribution normalization; see
\cite{Pless1963,Delsarte1973,MacWilliamsSloane1977}.  We record them because their exact
scaling is used in the equality arguments.

\begin{lemma}[First two dual moments]
\label{lem:moments}
For every \(N\)-point code \(\C\subseteq G^n\) with \(n\geq2\),
\begin{align}
 \sum_{w=1}^n\Adual_w
   &=\frac{q^n}{N}-1,\label{eq:mass}\\
 \sum_{w=1}^nw\Adual_w
   &=\frac{q^{n-1}}{N}\bigl(n(q-1)-A_1\bigr),\label{eq:first-moment}\\
 \sum_{w=1}^nw(w-1)\Adual_w
   &=\frac{q^{n-2}}{N}
   \bigl(n(n-1)(q-1)^2-2(n-1)(q-1)A_1+2A_2\bigr).
   \label{eq:second-moment}
\end{align}
\end{lemma}

\begin{proof}
Equation \eqref{eq:mass} is Parseval.  For the first moment, expand
\eqref{eq:dual} over ordered pairs \(z,z'\).  Character orthogonality gives
\[
 \sum_\xi\wt(\xi)\chi_\xi(z-z')
 =\begin{cases}
 n(q-1)q^{n-1},&z=z',\\
 -q^{n-1},&d(z,z')=1,\\
 0,&d(z,z')>1.
 \end{cases}
\]
Summing over pairs and dividing by \(N^2\) proves
\eqref{eq:first-moment}.

For \eqref{eq:second-moment}, sum over ordered pairs of distinct character
coordinates.  The corresponding character sum is nonzero only if
\(z-z'\) is supported on at most those two coordinates.  A diagonal pair
contributes \(n(n-1)(q-1)^2q^{n-2}\); a pair at distance one contributes
\(-2(n-1)(q-1)q^{n-2}\); and a pair at distance two contributes
\(2q^{n-2}\).  This proves the formula.
\end{proof}

We next record the classical Fourier form of Lloyd's theorem
\cite{Lloyd1957,Lenstra1972,Delsarte1973}.  By the
definition in Section~\ref{sec:prelim}, a code \(\C\) is a perfect code
correcting \(e\) errors if and only if
\begin{equation}
\label{eq:perfect-tiling}
 \one_\C*\one_{B(0,e)}=\one_{\X}.
\end{equation}

\begin{proposition}[Lloyd support]
\label{prop:lloyd-support}
If \(1\leq e\leq n-1\) and \(\C\) is a perfect code correcting \(e\)
errors, then
\[
 \Adual_w(\C)=0
 \quad\text{whenever}\quad
 1\leq w\leq n\ \text{and}\ c_w(e)=K_e^{(n-1,q)}(w-1)\neq0.
\]
Thus the dual support of a perfect code lies in the set of integral roots of
its Lloyd polynomial \(w\mapsto c_w(e)\).  Conversely, if
\(1\leq e\leq n-1\), \(|\C|V_e=q^n\), and \(\Adual_w(\C)=0\) whenever
\(1\leq w\leq n\) and \(c_w(e)\neq0\), then \(\C\) is a perfect code
correcting \(e\) errors.
\end{proposition}

\begin{proof}
Take the Fourier transform of \eqref{eq:perfect-tiling}.  At every nontrivial
character,
\[
 \wh{\one_\C}(\xi)c_{\wt(\xi)}(e)=0.
\]
The first claim follows from \eqref{eq:dual}; this is the Fourier form of
Lloyd's theorem in our normalization.  For the converse, \eqref{eq:dual}
gives \(\wh{\one_\C}(\xi)=0\) whenever \(\xi\) is nontrivial and
\(c_{\wt(\xi)}(e)\neq0\).  Hence the Fourier transform of
\(\one_\C*\one_{B(0,e)}\) vanishes at every nontrivial character and equals
\(|\C|V_e=q^n\) at the trivial one, and Fourier inversion gives
\eqref{eq:perfect-tiling}.
\end{proof}

\begin{remark}[Orthogonal-array and search constraints]
\label{rem:lloyd-oa-search}
Let \(\C\) be a perfect code correcting \(e\) errors, where
\(1\leq e\leq n-1\), and let \(\rho\) be the smallest positive
integer root of its Lloyd polynomial.  Proposition~\ref{prop:lloyd-support}
then gives
\[
 \wh{\one_\C}(\xi)=0
 \qquad\text{whenever}\qquad
 1\leq\wt(\xi)<\rho.
\]
By the standard Fourier characterization of orthogonal arrays, \(\C\) is
therefore an orthogonal array of strength \(\rho-1\)
\cite{Delsarte1973}.  Equivalently, for every \(J\subseteq[n]\) with
\(|J|<\rho\) and every \(a\in G^J\),
\[
 \bigl|\{z\in\C:z|_J=a\}\bigr|=\frac{|\C|}{q^{|J|}},
\]
where \(z|_J=(z_j)_{j\in J}\) denotes coordinate restriction.
In particular, \(q^{\rho-1}\mid |\C|\).  In the one- and two-error cases,
Lemmas~\ref{lem:one-perfect-parameters} and~\ref{lem:two-perfect-spectrum}
also determine the dual distribution of such a code from the parameters, and
hence, by MacWilliams inversion, its distance distribution.  These classical necessary conditions may be imposed as redundant
projection, Fourier, or distance-distribution constraints in computational
searches.
\end{remark}

For the single-error-correcting parameters, recall from
Section~\ref{sec:prelim} that
\[
 V=V_1=1+n(q-1),\qquad
 N=\frac{q^n}{V},\qquad
 k=\frac Vq.
\]
The Lloyd polynomial is \(c_w(1)=V-qw\).  The following standard consequence
of Lloyd support gives the dual distribution of a perfect code in our
normalization.  For \(n\geq2\), a perfect single-error-correcting code forces
the root \(k\) of its Lloyd polynomial to be an integer in
\(\{2,\ldots,n-1\}\), also when \(q\) is not a prime power.

\begin{lemma}[Dual distribution of a perfect single-error-correcting code]
\label{lem:one-perfect-parameters}
Let \(q\geq2\), \(n\geq2\), and suppose that
\(\C\subseteq G^n\) is a perfect single-error-correcting code.  Then
\[
 |\C|=N,\qquad q\mid V,\qquad k\in\{2,\ldots,n-1\}.
\]
Moreover,
\[
 \Adual_w(\C)=0\quad(w\neq0,k),\qquad
 \Adual_k(\C)=V-1.
\]
\end{lemma}

\begin{proof}
The tiling identity gives \(|\C|V=q^n\), hence \(|\C|=N\).
By Proposition~\ref{prop:lloyd-support}, the dual support lies in the set
of integral zeros of \(c_w(1)\).  On the other hand,
\eqref{eq:mass} gives
\[
 \sum_{w=1}^n \Adual_w(\C)=\frac{q^n}{N}-1=V-1>0.
\]
Thus \(c_w(1)=V-qw\) has an integral zero in
\(\{1,\ldots,n\}\).  The zero is \(k=V/q\), so \(q\mid V\).
Since \(V\equiv1-n\pmod q\), we have \(n=qr+1\) for some \(r\geq1\);
then \(k=1+r(q-1)\) and \(n-k=r\).  Hence \(1<k<n\).
Lloyd support now gives vanishing away from \(k\), and the mass identity gives
\(\Adual_k(\C)=V-1\).
\end{proof}

\subsection{The negative-slope criterion}

The extremal argument does not require global monotonicity of the spectral
weight.  Strict convexity and one negative adjacent slope at the Lloyd root
are enough.  We isolate this common step before proving the slope separately
in the three alphabet regimes.

\begin{samepage}
\begin{theorem}[The negative-slope criterion]
\label{thm:hamming-perfect}
Let \(q\geq2\) and \(n\geq2\), and assume that the parameters \(N\) and \(k\)
defined above are integers.  If
\begin{equation}
\label{eq:lloyd-slope}
 W_{n,q}(k)<W_{n,q}(k-1),
\end{equation}
then every \(N\)-point code \(\C\subseteq G^n\) satisfies
\begin{equation}
\label{eq:hamming-perfect-value}
 \Dtwo(\C)\geq\frac{V-1}{q^n}W_{n,q}(k).
\end{equation}
Equality holds if and only if \(\C\) is a perfect single-error-correcting
code.
\end{theorem}
\end{samepage}

\begin{proof}
The integrality of \(k\), together with
\[
 V\equiv1-n\pmod q,
\]
gives \(n=qr+1\) and \(k=1+r(q-1)\) for some \(r\geq1\).
Thus \(1<k<n\).

Let \(\C\) be any \(N\)-point code.  Strict discrete convexity gives the
supporting affine function through the consecutive points \(k-1,k\):
\begin{equation}
\label{eq:support-line}
 W_{n,q}(w)\geq W_{n,q}(k)+\sigma(w-k),
 \qquad 1\leq w\leq n,
\end{equation}
where
\[
 \sigma\triangleq W_{n,q}(k)-W_{n,q}(k-1)<0
\]
by \eqref{eq:lloyd-slope}.  Since
\[
 \frac{q^{n-1}}{N}=\frac Vq=k,\qquad n(q-1)=V-1,
\]
Lemma~\ref{lem:moments} gives
\[
 \sum_{w=1}^n\Adual_w(\C)=V-1
\]
and
\[
 \sum_{w=1}^nw\Adual_w(\C)
   =k\bigl(V-1-A_1(\C)\bigr).
\]
Therefore
\begin{align*}
 \sum_{w=1}^n\Adual_w(\C)W_{n,q}(w)
 &\geq (V-1)W_{n,q}(k)\\
 &\quad+\sigma\left(
 k\bigl(V-1-A_1(\C)\bigr)-k(V-1)\right)\\
 &=(V-1)W_{n,q}(k)-\sigma kA_1(\C)\\
 &\geq(V-1)W_{n,q}(k).
\end{align*}
Theorem~\ref{thm:spectral} proves \eqref{eq:hamming-perfect-value}.

It remains to identify equality.  Strict convexity shows that equality in
\eqref{eq:support-line} occurs only at \(w=k-1,k\).  Equality in the two
inequalities above first forces \(A_1(\C)=0\), because \(\sigma<0\), and then
forces the dual support into \(\{k-1,k\}\).  The mass and first-moment
identities become
\[
 \Adual_{k-1}+\Adual_k=V-1,\qquad
 (k-1)\Adual_{k-1}+k\Adual_k=k(V-1).
\]
Hence \(\Adual_{k-1}=0\), so the dual support is \(\{k\}\), the only zero of
\(c_w(1)=V-qw=q(k-w)\).  Since \(NV=q^n\), the
converse in Proposition~\ref{prop:lloyd-support} shows that \(\C\) is a
perfect single-error-correcting code.

Conversely, if \(\C\) is a perfect single-error-correcting code, then
Lemma~\ref{lem:one-perfect-parameters} and
Theorem~\ref{thm:spectral} give equality in
\eqref{eq:hamming-perfect-value}.
\end{proof}

The supporting line \eqref{eq:support-line} is a degree-one Delsarte
certificate, the analogue for the spectral weight of the linear auxiliary
function in \cite[Prop.~24]{CohnZhao2014}.

For alphabets of size at least four, the strict decrease in
Corollary~\ref{cor:convexity} gives \eqref{eq:lloyd-slope} immediately.
Consequently, Theorem~\ref{thm:hamming-perfect} proves
Theorem~\ref{thm:all-one-perfect} for \(q\geq4\).
\subsection{Slope inequalities for the two small alphabets}
\label{sec:small-alphabets}

For the two small alphabets, only the slope inequality \eqref{eq:lloyd-slope}
remains to be proved.  The spectral weight is not monotone for \(q=3\) (for
instance, \(W_{4,3}(1),\ldots,W_{4,3}(4)\) are \(245,26,14,20\)), and for
\(q=2\) it has an exact reflection symmetry; each case needs its own proof.

\paperparagraph{The ternary case}

When \(q=3\), integrality of \(k\) gives \(n=3r+1\) and \(k=2r+1\).
Thus the only missing input to Theorem~\ref{thm:hamming-perfect} is the
following inequality.

\begin{theorem}[Ternary Lloyd-root sign]
\label{thm:ternary-sign}
For every \(r\geq1\),
\begin{equation}
\label{eq:ternary-sign}
 W_{3r+1,3}(2r+1)<W_{3r+1,3}(2r).
\end{equation}
\end{theorem}

\begin{proof}
Lemma~\ref{lem:ternary-coefficient} writes
\(W_{3r+1,3}(2r)-W_{3r+1,3}(2r+1)\) as \(6\gamma_r\), and
Lemma~\ref{lem:ternary-recurrence} proves \(\gamma_r>0\) for every
\(r\geq1\) from the recurrence \eqref{eq:ternary-positive-recurrence}; its
telescoping certificate \eqref{eq:ternary-telescoper} is verified by computer
algebra, and the verification companion \cite{Zabokritskiy2026Verifier}
reproduces it symbolically.
\end{proof}

Theorem~\ref{thm:ternary-sign} verifies
\eqref{eq:lloyd-slope}, so Theorem~\ref{thm:hamming-perfect} proves
Theorem~\ref{thm:all-one-perfect} for \(q=3\).  Appendix
\ref{app:ternary-certificate} also shows that the discrepancy potential
\(f_{n,3}\) of Section~\ref{sec:universal-optimality} is not completely
monotonic.

\paperparagraph{The binary case}

The binary minimum and the closed form of the spectral weight are known from
Barg's analysis \cite{Barg2021}: the closed form is his Eq.~(56), which he
derives from the identity \cite[Eq.~(37)]{Barg2021}, attributed there to
Feinsilver and Fitzgerald \cite[Thm.~3.1.3]{FeinsilverFitzgerald1996}, and its
shape is
given in his Lemma~4.3.  The short symmetry argument below proves the slope
inequality in our \(q\)-ary normalization.
The exact specialization of the formulas is given in
Appendix~\ref{app:binary-formulas}.

\begin{lemma}[Known binary symmetry and Lloyd-root sign]
\label{lem:binary-sign}
For every \(n\geq1\) and \(1\leq w\leq n\),
\begin{equation}
\label{eq:binary-symmetry}
 W_{n,2}(w)=W_{n,2}(n+1-w).
\end{equation}
If \(n\geq2\) and \(k=(n+1)/2\) is an integer, then
\begin{equation}
\label{eq:binary-sign}
 W_{n,2}(k)<W_{n,2}(k-1).
\end{equation}
\end{lemma}

\begin{proof}
For \(q=2\), the two factors in \eqref{eq:W-integral} are
\[
 \alpha_2(\theta)=4\cos^2(\theta/2),\qquad
 \beta(\theta)=4\sin^2(\theta/2).
\]
The change of variables \(\theta\mapsto\pi-\theta\), interpreted modulo
\(2\pi\), interchanges them and proves \eqref{eq:binary-symmetry}.

If \(k=(n+1)/2\) is an integer and \(n\geq2\), then \(n=2k-1\geq3\).
Symmetry gives
\[
 W_{n,2}(k+1)=W_{n,2}(k-1).
\]
Strict convexity at \(k-1\) therefore gives
\[
 0<W_{n,2}(k+1)-2W_{n,2}(k)+W_{n,2}(k-1)
   =2\bigl(W_{n,2}(k-1)-W_{n,2}(k)\bigr),
\]
which is \eqref{eq:binary-sign}.
\end{proof}

\begin{proof}[Completion of the proof of
Theorem~\ref{thm:all-one-perfect}]
The case \(n=1\) was handled at the beginning of
Section~\ref{sec:hamming}.  For \(n\geq2\), the case \(q\geq4\) follows from
Corollary~\ref{cor:convexity}, and the case \(q=3\) follows from
Theorem~\ref{thm:ternary-sign}.  If \(q=2\), then \(k=(n+1)/2\), and
Lemma~\ref{lem:binary-sign} verifies \eqref{eq:lloyd-slope}.
Theorem~\ref{thm:hamming-perfect} now gives the bound and equality
characterization in every remaining case.
\end{proof}

\section{Perfect codes correcting multiple errors}
\label{sec:golay}

Theorem~\ref{thm:all-one-perfect} settles the one-error case for every
alphabet.  When a perfect code corrects more than one error, its Lloyd
polynomial has several roots, and Proposition~\ref{prop:lloyd-support} allows
its dual support to contain all of them; an affine minorant of \(W_{n,q}\)
need not agree with \(W_{n,q}\) at every Lloyd root.  For two-error-correcting
parameters, a universal-optimality criterion settles all alphabets of size at
least four.  The sphere-packing and integral-root conditions reduce the two
smallest alphabets to the binary repetition and ternary Golay parameters,
which are treated directly.  A quartic certificate for the binary Golay code and an
elementary argument for the full binary repetition family then complete the
remaining perfect-code cases.

\subsection{Complete monotonicity of the discrepancy potential}
\label{sec:universal-optimality}

We first prove the distance-kernel property needed for the large-alphabet
two-error theorem.  For a function \(u:\{0,\ldots,n\}\to\mathbb R\), write
\(\Delta u(w)\triangleq u(w+1)-u(w)\).  Following \cite{CohnZhao2014}, we call
\(u\) completely monotonic if \((-1)^j\Delta^ju(w)\geq0\) whenever
\(0\leq j\leq n-w\).  Within a fixed Hamming space and cardinality, a code is
universally optimal \cite[Def.~1]{CohnZhao2014} if it minimizes \(\sum_{w=1}^n A_wu(w)\) for every completely monotonic potential
\(u\).

\begin{proposition}[Completely monotonic discrepancy potential]
\label{prop:complete-monotone}
For \(q\geq4\), the potential
\[
 f_{n,q}(w)\triangleq\lambda(n)-\lambda(w),\qquad 0\leq w\leq n,
\]
is completely monotonic:
\[
 (-1)^j\Delta^jf_{n,q}(w)\geq0
 \qquad(0\leq w,\ 0\leq j\leq n-w),
\]
strictly so when \(j\geq1\).  Consequently, every universally optimal
\(q\)-ary code minimizes \(\Dtwo\) among codes of the same cardinality.
\end{proposition}

\begin{proof}
Put \(p_w\triangleq q^{-w}T_w(q)\).  Fourier inversion gives
\[
 p_w=\frac1{2\pi}\int_0^{2\pi}
 \left(1-\frac4q\sin^2\frac{\theta}{2}\right)^w\,d\theta.
\]
For \(q\geq4\) the base lies in \([0,1]\), and hence
\[
 (-1)^j\Delta^jp_w
 =\frac1{2\pi}\int_0^{2\pi}
 \left(1-\frac4q\sin^2\frac{\theta}{2}\right)^w
 \left(\frac4q\sin^2\frac{\theta}{2}\right)^j\,d\theta\geq0,
\]
with strict inequality for \(j\geq1\).  Proposition~\ref{prop:lambda} gives
\[
 \Delta f_{n,q}(w)=-q^{n-1}p_w,
\]
so all higher alternating differences have the required sign; order zero
follows from the monotonicity of \(\lambda\).

With \(E_f(\C)\triangleq\sum_{w=1}^nA_w(\C)f_{n,q}(w)\),
\eqref{eq:invariance-distance} becomes
\[
 \Dtwo(\C)=\frac1N E_f(\C)+\Lambda_{n,q}
            -\left(1-\frac1N\right)\lambda(n).
\]
The last two terms are fixed, proving the final assertion by universal
optimality \cite{CohnZhao2014}.
\end{proof}

We next fix the quasicode notation used in the two-error proof.  Let \(K\) be the Krawtchouk matrix
in \(\mathbb R^{(n+1)\times(n+1)}\), with rows and columns indexed by
\(i,j\in\{0,\ldots,n\}\) and entries
\(K_{ij}\triangleq K_i^{(n,q)}(j)\).  Thus \(K^2=q^nI\); that is, the
MacWilliams transform is an involution in this normalization
\cite{Delsarte1973,MacWilliamsSloane1977}.
For an \(N=q^n/V\)-point code, its dual distribution \(Q\triangleq\Adual\) is a
quasicode of size \(V\) in the sense of
\cite[Def.~6]{CohnZhao2014}.  Explicitly,
\[
 Q\geq0,\qquad Q_0=1,\qquad \sum_{i=0}^nQ_i=V,
 \qquad KQ=VA\geq0;
\]
neither integrality nor realizability of \(Q\) as an actual code is required.
Theorem~\ref{thm:spectral} reads
\[
 \Dtwo(\C)=q^{-n}\sum_{w=1}^nQ_wW_{n,q}(w).
\]

For \(q\geq4\), the one-error bound \eqref{eq:main-bound} also follows from
the universal lower energy bound of \cite[Thm.~4.2]{BoyvalenkovEtAl2017}, as
follows; no proof in this paper depends on this derivation.  By
\eqref{eq:W-integral},
\[
 W_{n,q}(x+1)=\frac1{2\pi}\int_0^{2\pi}
 \alpha_q(\theta)^{n-1}
 \left(\frac{\beta(\theta)}{\alpha_q(\theta)}\right)^x d\theta,
 \qquad x\geq0,
\]
where \(0\leq\beta/\alpha_q\leq1\).  Since
\cite[Thm.~4.2]{BoyvalenkovEtAl2017} concerns absolutely monotone potentials on
\([-1,1)\), put \(\rho\triangleq\beta/\alpha_q\).  For
\(0<\varepsilon<4/(q-2)^2\), define
\[
 W_\varepsilon(w)\triangleq\frac1{2\pi}
 \int_{\{\theta:\rho(\theta)\geq\varepsilon\}}
 \alpha_q(\theta)^{n-1}\rho(\theta)^{w-1}\,d\theta,
 \qquad w>0.
\]
As a function of the real variable \(w\in(0,n]\), this positive function is
completely monotonic, that is, \((-1)^jW_\varepsilon^{(j)}\geq0\) for all
\(j\); so in the inner-product variable \(t=1-2w/n\) of
\cite[Eq.~(1)]{BoyvalenkovEtAl2017}, \(h_\varepsilon(t)\triangleq
W_\varepsilon(n(1-t)/2)\) is absolutely monotonic on \([-1,1)\).  At
admissible one-error parameters with \(n\geq2\), we have \(V>q=R(n,1)\), so
\(V\) lies in the range \((R(n,1),R(n,2)]\) of
\cite[Thm.~4.2]{BoyvalenkovEtAl2017}, and the linear-programming certificate of
that theorem at the strength-two Rao endpoint applies to the quasicode \(Q\)
and gives
\[
 \sum_{w=1}^nQ_wW_\varepsilon(w)\geq(V-1)W_\varepsilon(k).
\]
At each integer \(w\geq1\), monotone convergence gives
\(W_\varepsilon(w)\uparrow W_{n,q}(w)\) as \(\varepsilon\downarrow0\).
Passing to this limit on both sides yields
\[
 \sum_{w=1}^nQ_wW_{n,q}(w)\geq(V-1)W_{n,q}(k),
\]
which is exactly \eqref{eq:main-bound}.  Unlike this derivation, the direct
criterion of Section~\ref{sec:hamming} shows that equality forces a perfect
tiling and covers \(q=3\), where \(\beta/\alpha_3>1\) near \(\theta=\pi\) and
the discrepancy potential \(f_{n,3}\) is not completely monotonic;
Appendix~\ref{app:ternary-certificate} records the first failed finite
difference of \(f_{n,3}\).

\subsection{Proof of the two-error theorem}

We now prove Theorem~\ref{thm:two-perfect-all-q}.  Recall from
Section~\ref{sec:prelim} the quantities \(V^{(2)},N^{(2)},L_2,r,s,\bar w\),
the formal masses \(b_r,b_s\) in \eqref{eq:formal-two-masses}, and the
benchmark \(\delta_{n,q}^{(2)}\) in \eqref{eq:radius-two-benchmark}.  The next
lemma identifies the benchmark as the discrepancy of every perfect code at
these parameters; it proves the converse direction for every \(q\geq2\).

\begin{lemma}[Dual distribution of a perfect two-error-correcting code]
\label{lem:two-perfect-spectrum}
Let \(q\geq2\) and \(n\geq3\), and suppose that \(L_2\) has two distinct
integral roots \(r<s\) in \(\{1,\ldots,n\}\).  If \(\C\subseteq G^n\) is a
perfect two-error-correcting code, then \(|\C|=N^{(2)}\),
\[
 \Adual_w(\C)=0\quad(w\neq0,r,s),\qquad
 \Adual_r(\C)=b_r,\qquad \Adual_s(\C)=b_s,
\]
and \(\Dtwo(\C)=\delta_{n,q}^{(2)}\).
\end{lemma}

\begin{proof}
The tiling identity gives \(|\C|V^{(2)}=q^n\), hence \(|\C|=N^{(2)}\).  The
quadratic \(L_2(w)=c_w(2)\) has no zeros in \(\{1,\ldots,n\}\) other than
\(r,s\), so Proposition~\ref{prop:lloyd-support} gives \(\Adual_w(\C)=0\)
for \(w\neq0,r,s\).  Distinct codewords are at distance at least five, so
\(A_1=0\), and Lemma~\ref{lem:moments} gives
\[
 \Adual_r+\Adual_s=V^{(2)}-1,\qquad
 r\Adual_r+s\Adual_s=\frac{q^{n-1}}{N^{(2)}}\,n(q-1)=V^{(2)}\bar w.
\]
By \eqref{eq:formal-two-masses}, the unique solution is \((b_r,b_s)\).
Theorem~\ref{thm:spectral} and \eqref{eq:radius-two-benchmark} give
\(\Dtwo(\C)=\delta_{n,q}^{(2)}\).
\end{proof}

We first treat \(q\geq4\), where the discrepancy potential is completely
monotonic.  In this branch the lower bound is also the universal lower energy
bound of \cite[Thm.~4.2]{BoyvalenkovEtAl2017} at the strength-four Rao
endpoint; no proof in this paper depends on the following sketch.  Indeed, if
\(A\) is the distance distribution of an \(N^{(2)}\)-point code, then
\(Q\triangleq\Adual\) satisfies the quasicode constraints displayed above with
size \(V^{(2)}=R(n,4)>R(n,3)\).  The
linear-programming certificate underlying \cite[Thm.~4.2]{BoyvalenkovEtAl2017}
depends only on Delsarte feasibility and therefore applies unchanged to \(Q\).
At this common endpoint of its degree-three and degree-four quadratures, the
quadrature nodes are the Lloyd roots \(r,s\), and the certificate yields the
value \(\delta_{n,q}^{(2)}\).  The cutoff-and-limit argument of
Section~\ref{sec:universal-optimality} applies here because both Lloyd roots
are positive integers.  For perfect codes with these parameters,
binomial-moment extremality and related universal-optimality statements
follow from \cite{AshikhminBarg1999} and \cite[Props.~9 and 23]{CohnZhao2014}.
Since \(f_{n,q}\) is strictly completely monotonic, these results also give
the equality case of Theorem~\ref{thm:two-perfect-all-q} when such a code
exists: every \(N^{(2)}\)-point code of minimal discrepancy then has all
binomial moments, and hence the distance distribution, of the perfect code
\cite[Lemma~4 and the discussion after Prop.~8]{CohnZhao2014}, so its minimum
distance is at least five and it is perfect.  The proof below instead
compares directly with the
formal inverse transform \(A^*\) defined next
(Lemma~\ref{lem:formal-two-comparison}).

Put \(a\triangleq q-1\).  For every \(q\geq2\), the centered Lloyd polynomial is
\begin{equation}
\label{eq:lloyd-two-centered}
 2L_2(\bar w+y)=q^2y^2+q(q-4)y+2-an.
\end{equation}
Define the formal dual vector \(Q^*\in\mathbb R^{n+1}\), indexed by character
weights \(w\in\{0,\ldots,n\}\), by
\[
 Q_0^*\triangleq1,\qquad Q_r^*\triangleq b_r,\qquad Q_s^*\triangleq b_s,
\]
with all other coordinates zero.  Using the Krawtchouk matrix \(K\), put
\begin{equation}
\label{eq:formal-A-star}
 A^*\triangleq\frac1{V^{(2)}}KQ^*.
\end{equation}
Thus \(A^*\in\mathbb R^{n+1}\) is indexed by Hamming distance.  It is a
formal, possibly signed distance-distribution vector.
Lemma~\ref{lem:two-node-quadrature}, the Rao-endpoint quadrature in our
normalization, gives
\begin{equation}
\label{eq:formal-two-design}
 1<r<\bar w<s<n,\qquad b_r,b_s>0,
\end{equation}
and
\begin{equation}
\label{eq:formal-two-design-moments}
 A_0^*=1,\qquad
 A_i^*=0\quad(1\leq i\leq4),\qquad
 \sum_{i=0}^nA_i^*=N^{(2)}.
\end{equation}

\begin{lemma}[Rao-endpoint comparison in distance normalization]
\label{lem:formal-two-comparison}
Assume \(q\geq4\) and the hypotheses of
Theorem~\ref{thm:two-perfect-all-q}.  Let \(A\) be the distance distribution
of an \(N^{(2)}\)-point code, and let \(A^*\) be defined by
\eqref{eq:formal-A-star}.  Then, for every completely monotonic potential
\(f:\{0,1,\ldots,n\}\to\mathbb R\),
\begin{equation}
\label{eq:formal-energy-comparison}
 f^{\mathsf T}A\geq f^{\mathsf T}A^*.
\end{equation}
No feasibility or nonnegativity assumption on \(A^*\) is required.
\end{lemma}

\begin{proof}
This is the pair-covering certificate underlying
\cite[Prop.~21]{CohnZhao2014} at the strength-four Rao endpoint, written in
the present normalization.  We transform \(f\) to the dual side, bound the
transform from below by a cubic that agrees with it on the nonconstant
support \(\{r,s\}\) of \(Q^*\) and has
nonnegative nonconstant Krawtchouk coefficients, compare \(Q\) with
\(Q^*\), and transform back.
Let \(Q\triangleq\Adual\).  Since \(\sum_{i=0}^nA_i=N^{(2)}\), the transform conventions give
\[
 KQ=V^{(2)}A,
 \qquad Q_0=1,
 \qquad Q_i\geq0.
\]
Following \cite[Def.~11]{CohnZhao2014}, a radial function on the Hamming
space is called positive definite when all coefficients in its Krawtchouk
expansion are nonnegative.  The results of \cite{CohnZhao2014} used below
concern functions on \(\{0,1,\ldots,n\}\), with the Krawtchouk polynomials
\eqref{eq:kraw} and with the condition of order zero included in complete
monotonicity, as here.
Put \(g\triangleq K^{\mathsf T}f\).  Complete monotonicity is preserved by this
transform \cite[Lemma~10]{CohnZhao2014}.  Let \(h\) be the cubic polynomial
interpolating \(g\) at
\[
 r-1,\quad r,\quad s,\quad s+1.
\]
The set \(T\triangleq\{r-1,r,s,s+1\}\) contains the nonconstant support
\(\{r,s\}\) of \(Q^*\) and, in increasing order, consists of two consecutive pairs;
thus it is a pair covering in the terminology of \cite{CohnZhao2014}.
Lemma~19 there,
applied in the reverse node order \(s+1,s,r,r-1\), gives
\[
 h(i)\leq g(i)\qquad(1\leq i\leq n)
\]
and a Newton expansion with nonnegative coefficients in the functions
\[
 1,\qquad s+1-x,\qquad
 (s+1-x)(s-x),\qquad
 (s+1-x)(s-x)(r-x).
\]
Each nonconstant function in this list is positive definite.  Indeed,
\(s>\bar w=(q-1)n/q\) by \eqref{eq:formal-two-design}, so the first two of
them are covered by
Lemmas~12--13 and Corollary~14 of \cite{CohnZhao2014}, while
\[
 (s+1-x)(s-x)(r-x)
 =\frac{2}{q^2}(s+1-x)L_2(x)
\]
is a product of positive-definite functions because
\(L_2=K_0^{(n,q)}+K_1^{(n,q)}+K_2^{(n,q)}\).  Therefore
\[
 h(x)=\sum_{j=0}^3\eta_jK_j^{(n,q)}(x)
 \qquad\text{with}\qquad \eta_j\geq0\quad(j\geq1).
\]

Using \(KQ=V^{(2)}A\), the nonnegativity of the distance distribution, and
\(Q_0=1\), we obtain
\begin{align*}
 \sum_{i=1}^ng(i)Q_i
 &\geq\sum_{i=1}^nh(i)Q_i\\
 &=V^{(2)}\sum_{j=0}^3\eta_jA_j-h(0)\\
 &\geq V^{(2)}\eta_0-h(0).
\end{align*}
For \(Q^*\), the only nonzero coordinates among \(1,\ldots,n\) are
\(r\) and \(s\), where \(h=g\), and
\eqref{eq:formal-two-design-moments} gives
\[
 \sum_{i=1}^ng(i)Q_i^*
 =\sum_{i=1}^nh(i)Q_i^*
 =V^{(2)}\eta_0-h(0).
\]
After adding the common value at zero,
\[
 g^{\mathsf T}Q\geq g^{\mathsf T}Q^*.
\]
Finally,
\[
 A=\frac1{V^{(2)}}KQ,
 \qquad
 A^*=\frac1{V^{(2)}}KQ^*,
\]
so
\[
 f^{\mathsf T}A
 =\frac1{V^{(2)}}(K^{\mathsf T}f)^{\mathsf T}Q
 \geq
 \frac1{V^{(2)}}(K^{\mathsf T}f)^{\mathsf T}Q^*
 =f^{\mathsf T}A^*.\qedhere
\]
\end{proof}

\begin{proof}[Proof of Theorem~\ref{thm:two-perfect-all-q} for \(q\geq4\)]
Let \(A\) be the distance distribution of an arbitrary
\(N^{(2)}\)-point code.  Apply Lemma~\ref{lem:formal-two-comparison} to the
strictly completely monotonic potential
\[
 f_{n,q}(w)=\lambda(n)-\lambda(w).
\]
Using \eqref{eq:invariance-distance}, \eqref{eq:formal-energy-comparison},
\(\sum_iA_i=N^{(2)}=\sum_iA_i^*\) from \eqref{eq:formal-two-design-moments},
and \(\lambda(0)=0\), we obtain
\[
 \Dtwo(\C)
 \geq
 \Lambda_{n,q}-\frac1{N^{(2)}}\sum_{i=0}^nA_i^*\lambda(i).
\]
Substituting
\(A^*=(V^{(2)})^{-1}KQ^*\) and the Krawtchouk expansion of \(\lambda\), whose
coefficients \(\wh\lambda_0^{(n,q)}=\Lambda_{n,q}\) and
\(\wh\lambda_j^{(n,q)}=-q^{-n}W_{n,q}(j)\) \((j\geq1)\) are given by
\eqref{eq:W-lambda-coefficient}, and using \(K^2=q^nI\) and
\(N^{(2)}V^{(2)}=q^n\), we obtain
\[
 \frac1{N^{(2)}}\sum_{i=0}^nA_i^*\lambda(i)
 =\sum_{j=0}^nQ_j^*\wh\lambda_j^{(n,q)}
 =\Lambda_{n,q}-\frac1{q^n}\sum_{j=1}^nQ_j^*W_{n,q}(j).
\]
Hence the right-hand side of the preceding bound equals
\[
 \frac1{q^n}\sum_{j=1}^nQ_j^*W_{n,q}(j)
 =\delta_{n,q}^{(2)}.
\]
This proves the bound.

We now characterize equality.  On \(\{1,\ldots,n\}\), Newton expansion at
\(n\) gives
\begin{equation}
\label{eq:strict-fundamental-expansion}
 f_{n,q}(w)
 =\sum_{j=1}^{n-1}\alpha_j\binom{n-w}{j},
 \qquad
 \alpha_j\triangleq(-1)^j\Delta^jf_{n,q}(n-j)>0,
\end{equation}
where strict positivity follows from
Proposition~\ref{prop:complete-monotone}.  For
\(1\leq j\leq n-1\), put
\[
 \phi_j(w)\triangleq\binom{n-w}{j}.
\]
Each \(\phi_j\) is completely monotonic, so
Lemma~\ref{lem:formal-two-comparison} gives
\[
 \phi_j^{\mathsf T}(A-A^*)\geq0.
\]
If equality holds for \(f_{n,q}\), then the positive coefficients in
\eqref{eq:strict-fundamental-expansion} force equality in every one of these
binomial moments.  The missing zeroth moment is fixed as well:
\[
 \sum_{i=1}^nA_i=N^{(2)}-1
 =\sum_{i=1}^nA_i^*.
\]
The functions
\[
 1,\binom{n-w}{1},\ldots,\binom{n-w}{n-1}
\]
are the binomial-moment basis underlying
\cite{AshikhminBarg1999} and \cite[Lemma~4]{CohnZhao2014}, and form a basis
for all functions on \(\{1,\ldots,n\}\).  Explicitly, after
ordering the points as \(n,n-1,\ldots,1\), their evaluation matrix is the
Pascal matrix
\[
 \left(\binom{x}{j}\right)_{0\leq x,j\leq n-1},
\]
which is triangular with diagonal entries one.  Hence all coordinates agree:
\[
 A=A^*.
\]
By \eqref{eq:formal-two-design-moments},
\(A_1=\cdots=A_4=0\).  Thus distinct codewords are at distance
at least five, so their radius-two balls are disjoint.  Since
\[
 |\C|V^{(2)}=q^n,
\]
they partition the space, and \(\C\) is perfect.

Conversely, every perfect two-error-correcting code attains
\(\delta_{n,q}^{(2)}\) by Lemma~\ref{lem:two-perfect-spectrum}.
\end{proof}

For \(f=f_{n,q}\), \eqref{eq:W-lambda-coefficient} and \(K^2=q^nI\) give
\(K^{\mathsf T}f=W_{n,q}\) on \(\{1,\ldots,n\}\), so the cubic \(h\) in the
proof of Lemma~\ref{lem:formal-two-comparison} is then a minorant of the
spectral weight that agrees with it at \(r-1,r,s,s+1\), the degree-three
counterpart of the supporting line \eqref{eq:support-line}.

The vector \(A^*\) is computable from \(q,n\).  A negative coordinate of
\(A^*\), or a coordinate with \(N^{(2)}A_i^*\notin\mathbb Z\), rules out a
perfect code.

\subsection{The two small alphabets}

We begin with the ternary Golay parameters:
length \(11\), cardinality \(3^6\), and minimum distance \(5\).  These are the
parameters of a perfect two-error-correcting code, and a quadratic minorant,
evaluated with the dual moments of Lemma~\ref{lem:moments}, suffices.

Fix a ternary perfect code
\(\mathcal G\subseteq\F_3^{11}\) with parameters \((11,3^6,5)\).
Although ternary Golay codes are universally optimal for completely
monotonic energies \cite[Table~I and Props.~9 and 23]{CohnZhao2014}, that
result does not imply the next theorem, because the ternary discrepancy
potential is not completely monotonic.

\begin{theorem}[Ternary Golay perfect codes]
\label{thm:ternary-golay}
Every code \(\C\subseteq\F_3^{11}\) of size \(3^6\) satisfies
\[
 \Dtwo(\C)\geq\Dtwo(\mathcal G).
\]
Equality holds if and only if \(\C\) is a perfect two-error-correcting code.
\end{theorem}

\begin{proof}
Proposition~\ref{prop:golay-certificate} proves the bound in the stronger
form \eqref{eq:golay-bound}; its proof constructs a quadratic minorant of
\(W_{11,3}\) that agrees with the spectral weight at both Lloyd roots.  The
first two dual moments evaluate this minorant and give a sharp lower bound.
In the equality case, the same moment identities eliminate the one additional
contact point of the minorant, so the dual support is exactly the set of
Lloyd roots.  The converse in Proposition~\ref{prop:lloyd-support} then shows
that \(\C\) is a perfect two-error-correcting code.  Finally,
Lemma~\ref{lem:two-perfect-spectrum} shows that every perfect code with these
parameters attains the bound.
\end{proof}

Every two-point perfect code is necessarily a binary repetition code of odd
length.  Indeed, if \(q\geq3\), one can choose a word whose coordinate in
each position differs from the corresponding coordinates of both codewords;
its distance to each codeword is \(n\), whereas the two disjoint balls have
radius less than \(n\), so the word lies in neither ball.  Hence \(q=2\).  Disjointness of the two radius-\(e\) balls gives
distance at least \(2e+1\) between their centers, while the sphere-packing
identity
\[
 2\sum_{j=0}^{e}\binom nj=2^n
\]
and symmetry of the binomial coefficients force \(n=2e+1\).  The centers
must therefore be antipodal.

The perfect binary codes of size two are covered by Barg's repetition-code
bound \cite[Cor.~5.3, Eq.~(68), and Example~1]{Barg2021}.  We record the
exact equality form needed in the proof of Theorem~\ref{thm:all-perfect}.
For odd \(n\), put
\[
 \mathcal R_n\triangleq\{0^n,1^n\}\subseteq\F_2^n.
\]

\begin{proposition}[Binary repetition codes]
\label{prop:repetition}
Let \(e\geq1\) and \(n=2e+1\).  Every two-point code
\(\C\subseteq\F_2^n\) satisfies
\[
 \Dtwo(\C)\geq\Dtwo(\mathcal R_n).
\]
Equality holds if and only if the two codewords are antipodal, equivalently,
if and only if \(\C\) is a perfect \(e\)-error-correcting code.
\end{proposition}

\begin{proof}
This is the cited repetition-code bound in the normalization of
Theorem~\ref{thm:invariance}; the following calculation identifies its
equality case directly.
If the two codewords have distance \(d\), their distance distribution has
\(A_d=1\) and no other nonzero entry away from \(A_0\).  Hence
\eqref{eq:invariance-distance} gives
\[
 \Dtwo(\C)=\Lambda_{n,2}-\frac12\lambda(d).
\]
The binary increment formula in Proposition~\ref{prop:lambda} shows that,
when \(n\) is odd, \(\lambda(d)\) has its unique maximum at \(d=n\).
Thus equality holds exactly for an antipodal pair.  Its two radius-\(e\)
balls partition \(\F_2^n\).
\end{proof}

The next lemma is the classical Ramanujan--Nagell and Ljunggren arithmetic
reduction specialized to the sphere-packing and integral-root conditions.
Bennett \cite[Eqs.~(1)--(7)]{Bennett2026} writes the same conditions in terms
of the two roots.

\begin{lemma}[Admissible two-error parameters over the two smallest alphabets]
\label{lem:small-q-two-parameters}
Let \(q\in\{2,3\}\) and \(n\geq3\).  Suppose that \(N^{(2)}\) is an integer
and that \(L_2\) has two distinct integral roots \(r<s\) in
\(\{1,\ldots,n\}\).  Then
\[
 (q,n,N^{(2)},r,s)
 =
 (2,5,2,2,4)
 \quad\text{or}\quad
 (3,11,3^6,6,9).
\]
\end{lemma}

\begin{proof}
See Appendix~\ref{app:two-error-details}.
\end{proof}

It remains to prove Theorem~\ref{thm:two-perfect-all-q} for the two smaller
alphabets.
\begin{proof}[Proof for \(q=2,3\)]
For \(q=2,3\), Lemma~\ref{lem:small-q-two-parameters} reduces the parameters
to those of the length-five binary repetition code and of the ternary Golay
code, respectively.  Proposition~\ref{prop:repetition} and
Theorem~\ref{thm:ternary-golay} give the required sharp inequalities and
equality characterizations.

Their right-hand sides are discrepancies of perfect two-error-correcting
codes at these parameters, and Lemma~\ref{lem:two-perfect-spectrum}
identifies them with \(\delta_{n,q}^{(2)}\).
\end{proof}

\subsection{Codes correcting at least three errors}

We now turn to codes correcting at least three errors.  Beyond the
repetition family, the only perfect codes with packing radius at least three
and more than two codewords have the binary Golay parameters.  Their
certificate needs a quartic minorant, one degree more than the cubic Lloyd
polynomial (Remark~\ref{rem:binary-golay-cubic}).  Fix a binary perfect code
\(\mathcal G_2\subseteq\F_2^{23}\) with parameters \((23,2^{12},7)\).

\begin{theorem}[Binary Golay perfect codes]
\label{thm:binary-golay}
Every code \(\C\subseteq\F_2^{23}\) of size \(2^{12}\) satisfies
\begin{equation}
\label{eq:binary-golay-main}
 \Dtwo(\C)\geq
 \frac{409\,732\,557}{1\,048\,576}
 =\Dtwo(\mathcal G_2).
\end{equation}
Equality holds if and only if \(\C\) is a perfect three-error-correcting code.
\end{theorem}

\begin{proof}
Appendix~\ref{app:binary-golay-certificate} gives an explicit quartic
minorant \(p\) of \(W_{23,2}\).  It agrees with the spectral weight exactly at
the three Lloyd roots and lies strictly below it at every other integer
weight.  The first four factorial dual moments evaluate
\(2^{-23}\sum_{w=1}^{23}\Adual_wp(w)\) as the constant in
\eqref{eq:binary-golay-main} plus positive multiples of \(A_1+A_2\) and
\(A_3+A_4\).  Equality therefore forces the dual support into the set of
Lloyd roots, and the converse in Proposition~\ref{prop:lloyd-support} shows
that \(\C\) is a perfect three-error-correcting code.  Conversely, minimum
distance seven and Lloyd support show that every binary Golay perfect code
attains the bound.
\end{proof}

Barg identifies the binary Golay code as a discrepancy minimizer by computer
\cite[Example~3(i) and Thm.~5.4]{Barg2021} and quotes the rounded value
\(390.75\) \cite[Sec.~3.2]{Barg2021}; the appendix supplies the exact rational
value, an explicit quartic certificate, and the equality characterization.

\begin{proof}[Proof of Theorem~\ref{thm:all-perfect}]
If \(|\mathcal P|=1\), the statement is immediate.  Let \(e\geq1\) be the
packing radius of \(\mathcal P\).  Since distinct codewords of \(\mathcal P\)
are at distance at least \(2e+1\), we have \(n\geq2e+1\).

If \(e=1\), Lemma~\ref{lem:one-perfect-parameters} shows that the
parameters of Theorem~\ref{thm:all-one-perfect} are integers, and that
theorem gives the result.  If \(e=2\), then \(n\geq5\), and the
sphere-packing condition and Lloyd's theorem
\cite{Lloyd1957,Lenstra1972,Delsarte1973} supply the remaining hypotheses of
Theorem~\ref{thm:two-perfect-all-q}, which gives the result.

Assume \(e\geq3\).  If \(|\mathcal P|=2\), then \(\mathcal P\) is an
odd-length binary repetition code by the argument before
Proposition~\ref{prop:repetition}, and that proposition applies.  Suppose therefore that \(|\mathcal P|\geq3\).  For prime-power
alphabet sizes, the parameter classification leaves only the binary Golay
parameters, covered by Theorem~\ref{thm:binary-golay}.  Over non-prime-power
alphabets, the nonexistence results of Reuvers, Best, and Hong rule out every
nontrivial case of packing radius at least three; see
\cite{Reuvers1977,Best1983,Hong1984} and
\cite[Table~1]{CazorlaGarcia2024}.  This exhausts all possibilities.
In every branch, the applicable theorem or proposition above identifies
\(\Dtwo(\mathcal P)\) as the lower bound for the fixed parameters and
characterizes equality precisely by perfectness.  This proves both assertions
of the theorem.
\end{proof}

The parameter classification applies to arbitrary, possibly nonlinear,
perfect codes.  A classification-free proof for every number of correctable
errors remains open; extending the spectral method would require a minorant
of \(W_{n,q}\) that agrees with it at the Lloyd roots, with a moment expansion
that bounds arbitrary competitors and equality conditions that force perfect
tiling.

\section{Discrepancy as a multiscale smoothing error}
\label{sec:interpretation}
\label{sec:smoothing}

Code smoothing asks whether adding noise to a uniformly random codeword
yields a nearly uniform distribution.
The general viewpoint for codes and lattices is developed in
\cite{DebrisAlazardEtAl2023}; the binary uniform-ball case and its
information-theoretic applications are studied in \cite{PathegamaBarg2023}.
Most work fixes one noise kernel, whereas total ball discrepancy aggregates
every uniform-ball kernel.  The next identity writes \(\Dtwo(\C)\) as a
weighted sum of these smoothing errors over all radii.  For a
fixed binary radius, \cite[Props.~4 and A1]{PathegamaBarg2023} express the
smoothing error through the distance and dual distance distributions; here the
corresponding identities (Theorems~\ref{thm:invariance}
and~\ref{thm:spectral}) are aggregated over all radii and extended to the
\(q\)-ary setting with the exact discrepancy weights.

Let
\[
 f_\C\triangleq\frac{\one_\C}{N},\qquad
 b_t\triangleq\frac{\one_{B(0,t)}}{V_t},\qquad
 u\triangleq\frac{\one_{\X}}{q^n}
\]
be the code, ball-noise, and uniform probability mass functions,
respectively.
The codeword proportion \(|\C\cap B(x,t)|/N\) in \eqref{eq:intro-discrepancy}
equals \(V_t(f_\C*b_t)(x)\).
For a function \(h:\X\to\mathbb C\), use the counting-measure norm
\[
 \|h\|_2^2\triangleq\sum_{x\in\X}|h(x)|^2.
\]
For probability mass functions \(P,Q\) with \(P\ll Q\), write
\[
 \chi^2(P\|Q)\triangleq
 \sum_{x\in\X}\frac{(P(x)-Q(x))^2}{Q(x)}
\]
for Pearson's chi-squared divergence.

\begin{proposition}[Multiscale smoothing identity]
\label{prop:smoothing}
For every code \(\C\subseteq G^n\),
\begin{align}
 \Dtwo(\C)
 &=\sum_{t=0}^nV_t^2\|f_\C*b_t-u\|_2^2\label{eq:smoothing-l2}\\
 &=\frac1{q^n}\sum_{t=0}^nV_t^2
       \chi^2(f_\C*b_t\|u).\label{eq:smoothing-chi}
\end{align}
If \(\C\) is a perfect \(e\)-error-correcting code, then \(f_\C*b_e=u\);
equivalently, the radius-\(e\) summand in
\eqref{eq:smoothing-l2} vanishes.
\end{proposition}

\begin{proof}
 The first identity is a rewriting of \eqref{eq:intro-discrepancy}.  Since
\(u(x)=q^{-n}\),
\[
 \chi^2(P\|u)=q^n\|P-u\|_2^2,
\]
which gives the second identity.  For a perfect code, divide
\eqref{eq:perfect-tiling} by \(NV_e=q^n\).
\end{proof}

For each radius, Parseval and the Fourier transform in the proof of
Theorem~\ref{thm:spectral} give
\[
 V_t^2\|f_\C*b_t-u\|_2^2=\frac1{q^n}\sum_{w=1}^n\Adual_w(\C)c_w(t)^2,
\]
the \(q\)-ary form of the fixed-radius dual-distribution identity of
\cite{PathegamaBarg2023};
summing over \(t\) recovers \eqref{eq:spectral-discrepancy}.

Each summand in \eqref{eq:smoothing-chi} is a monotone transform of the
order-two R\'enyi divergence, with the natural logarithm:
\[
 D_2(P\|Q)\triangleq
 \log\sum_{x\in\X}\frac{P(x)^2}{Q(x)},
 \qquad
 \chi^2(P\|u)=\exp\bigl(D_2(P\|u)\bigr)-1.
\]
Thus total discrepancy aggregates \(\exp(D_2)-1\) over all radii, not \(D_2\)
itself.  Theorem~\ref{thm:all-perfect} identifies the minimizers of this
aggregate at every perfect-code parameter set.

\section{Conclusions and open problems}
\label{sec:conclusion}

At every length and cardinality that admit a perfect code, the perfect codes
are exactly the minimizers of total quadratic ball discrepancy.  The result
relates exact tiling at one radius to uniformity aggregated over all radii,
with arbitrary subsets as competitors.  The explicit benchmarks for
correcting one or two errors also apply before existence is known: a code
attains a benchmark if and only if it is perfect.  For \(q\geq4\), known
universal energy bounds also give these benchmarks as lower bounds, and,
since the discrepancy potential is strictly completely monotonic, known
universal optimality results give the characterization of equality whenever a
perfect code exists \cite{CohnZhao2014}; we supply this complete
monotonicity, the exact \(q\)-ary formulas relating discrepancy to energy, and
a direct
proof of the characterization that does not assume existence.  The ternary
and binary alphabets require separate arguments.
At larger radii, the
global characterization also uses the known classification and nonexistence
results, the binary Golay certificate, and an elementary argument for the
binary repetition codes.

The smoothing identity gives a distributional meaning to the same finite
objective.  Turning the benchmarks into effective search or
nonexistence criteria remains a computational problem.  This matters over
non-prime-power alphabets, where the one-error case remains incompletely
classified and no admissible two-error parameter set is known, although the
growing arithmetic exclusions \cite{CazorlaGarcia2024,Bennett2026} have not
ruled out all such sets.

Grishin and Zabokritskiy showed that near-minimal discrepancy implies
near-perfect tiling at the one-error parameters.  Fix \(q,n,N\) for which a
one-error perfect code exists, and let \(\C\) be any \(N\)-point code.  Then at
most \(N^2\bigl(\Dtwo(\C)-\Dtwo(q,n,N)\bigr)\) ambient points lie in none of
the balls of radius one around \(\C\) or in more than one
\cite[Thm.~4.1 and Prop.~2.2]{GrishinZabokritskiy2026Stability}.  We close
with three questions suggested by the results.

\begin{enumerate}
\item \emph{Can the discrepancy gap guide searches for perfect codes?}
At admissible one- or two-error parameters at which the existence of a
perfect code is open, every \(N\)-point code has discrepancy at least the
benchmark, with equality exactly for perfect codes, so a perfect code exists
if and only if \(\Dtwo(q,n,N)\) equals the benchmark.  Equality forces the
classical necessary conditions for perfect codes: dual support in the set of
Lloyd roots, orthogonal-array conditions, and a prescribed distance
distribution.  A code either satisfies these conditions or fails them,
whereas the excess of \(\Dtwo(\C)\) over the benchmark is defined for every
\(N\)-point code \(\C\) and vanishes exactly at perfect codes.  Can this
excess, or lower bounds for it, serve as an objective or as a pruning bound in
searches based on exact cover \cite{KaskiOstergard2006,OstergardPottonen2009}
or integer programming, or certify that \(\Dtwo(q,n,N)\) exceeds the
benchmark?

\begin{samepage}
\item \emph{What minimizes discrepancy when no perfect code exists?}
Fix \(q,n,N\) with no perfect code but with an \(N\)-point code correcting at
least one error.  Let \(e\) be the maximum number of errors corrected by such
codes.  Must some minimizer of \(\Dtwo(q,n,N)\) also correct \(e\) errors?
For \(q\geq4\) the answer is yes whenever a universally optimal code of size
\(N\) exists, since such a code minimizes the discrepancy
(Proposition~\ref{prop:complete-monotone}) and has the largest minimum
distance among \(N\)-point codes \cite[Sec.~I]{CohnZhao2014}.
\end{samepage}

\item \emph{Can the spectral method be extended to codes correcting more errors?}
For \(e\geq3\), the proof of our global characterization theorem relies in part on
known classification and nonexistence results.  Can one instead construct
from the parameters a polynomial minorant of the spectral weight \(W_{n,q}\)
of the Fourier formula that agrees with it at the integral Lloyd roots and whose equality conditions
force perfect tiling?
What degree is required?
\end{enumerate}

\appendix

\section{Technical details for the two-error theorem}
\label{app:two-error-details}

We record the quadrature at the strength-four Rao endpoint in our
normalization and the classical arithmetic reduction used in
Theorem~\ref{thm:two-perfect-all-q}.  The quadrature is the specialization of
the Levenshtein formulas in \cite[Eqs.~(21)--(24)]{BoyvalenkovEtAl2017} to
this endpoint; the proof below verifies it, with its nodes and weights, in
Hamming-distance coordinates.

\begin{samepage}
\begin{lemma}[Rao-endpoint two-node quadrature]
\label{lem:two-node-quadrature}
Under the hypotheses of Theorem~\ref{thm:two-perfect-all-q} with \(q\geq4\),
\[
 1<r<\bar w<s<n,
 \qquad b_r,b_s>0.
\]
Moreover, with \(a\triangleq q-1\),
\begin{equation}
\label{eq:radius-two-quadrature}
 \frac{P(0)+b_rP(r)+b_sP(s)}{V^{(2)}}
 =\sum_{i=0}^n q^{-n}\binom ni a^iP(i)
 \qquad(\deg P\leq4).
\end{equation}
Consequently, the vectors \(Q^*\) and \(A^*\) defined in
\eqref{eq:formal-A-star} and immediately before it satisfy
\eqref{eq:formal-two-design-moments}.
\end{lemma}
\end{samepage}

\begin{proof}
We verify the cited endpoint quadrature in the present normalization.
Equation~\eqref{eq:lloyd-two-centered} gives
\(L_2(\bar w)<0\), while \(L_2(1)>0\) and \(L_2(n)>0\); hence the two roots
have the stated order.  Their separation satisfies
\[
 (s-r)^2=1+\frac{4(q-1)(n-2)}{q^2}>1,
\]
so integral roots obey \(s-r\geq2\); in particular \(n\geq5\).  Their sum and
\eqref{eq:lloyd-two-centered} give
\[
 s-\bar w
 =\frac{s-r}{2}-\frac{q-4}{2q}
 \geq\frac{q+4}{2q}>\frac1q
 \geq\frac{\bar w}{V^{(2)}-1}.
\]
The formulas in \eqref{eq:formal-two-masses} now show that both masses are
positive.

The two sides of \eqref{eq:radius-two-quadrature} agree on \(1\) and \(x\)
by the definitions of \(b_r,b_s\).  They also agree on \(L_2\), because
\[
 L_2=K_0^{(n,q)}+K_1^{(n,q)}+K_2^{(n,q)},\qquad
 L_2(0)=V^{(2)},\qquad L_2(r)=L_2(s)=0,
\]
and the binomial average of \(L_2\) is \(1\).  Thus they agree on every
quadratic.  Every polynomial of degree at most four has the form
\[
 P(x)=R(x)+xL_2(x)S(x),\qquad \deg R\leq2,\quad \deg S\leq1.
\]
The second term vanishes at \(0,r,s\).  Its binomial average vanishes after
using \(x\binom nx=n\binom{n-1}{x-1}\) and Krawtchouk orthogonality, since
\(L_2(x)=K_2^{(n-1,q)}(x-1)\).  This proves the quadrature identity.  Applying
it to \(K_0^{(n,q)},\ldots,K_4^{(n,q)}\) yields \(A_0^*=1\) and
\(A_1^*=\cdots=A_4^*=0\); evaluating at the constant vector gives
\(\sum_iA_i^*=N^{(2)}\).
\end{proof}

\begin{proof}[Proof of Lemma~\ref{lem:small-q-two-parameters}]
For \(q=2\), write \(V^{(2)}=2^m\).  The identity
\[
 8V^{(2)}=(2n+1)^2+7
\]
and the Ramanujan--Nagell theorem, which states that the positive solutions
of \(x^2+7=2^\ell\) are
\[
 (x,\ell)=(1,3),(3,4),(5,5),(11,7),(181,15),
\]
\cite{Nagell1961}, leave
\(n\in\{0,1,2,5,90\}\).  Under \(n\geq3\), the roots
\[
 \frac{n+1\pm\sqrt{n-1}}2
\]
exclude \(n=90\), so \(n=5\), \(N^{(2)}=2\), and \((r,s)=(2,4)\).

For \(q=3\), one has \(V^{(2)}=2n^2+1=3^m\), or
\[
 \frac{3^m-1}{2}=n^2.
\]
Ljunggren's theorem states that, for \(X>1\) and \(m>2\),
\((X^m-1)/(X-1)=Y^2\) has only the solutions
\((X,m,Y)=(7,4,20)\) and \((3,5,11)\)
\cite{Ljunggren1943}.  Taking \(X=3\) leaves \(m=5,n=11\) when
\(n\geq3\).  Hence \(N^{(2)}=3^6\), and
\eqref{eq:lloyd-two-centered} gives \((r,s)=(6,9)\).
\end{proof}

\section{The ternary Lloyd-root certificate}
\label{app:ternary-certificate}

This appendix supplies the exact coefficient certificate used in
Theorem~\ref{thm:ternary-sign}.  The verification companion
\cite{Zabokritskiy2026Verifier} reproduces the telescoping identity
\eqref{eq:ternary-telescoper} symbolically and the initial values in
Lemma~\ref{lem:ternary-recurrence} in exact arithmetic.
For \(r\geq1\), define
\[
 \gamma_r\triangleq
 [z^{3r}](z-1)^{4r-3}(2z^2+5z+2)^r.
\]

\begin{lemma}[Coefficient formula for the ternary slope]
\label{lem:ternary-coefficient}
For every \(r\geq1\),
\begin{equation}
\label{eq:ternary-coefficient}
 W_{3r+1,3}(2r)-W_{3r+1,3}(2r+1)=6\gamma_r.
\end{equation}
\end{lemma}

\begin{proof}
Set
\[
 C_r(z)\triangleq(1+2z)^r(1-z)^{2r-1}.
\]
Writing \(\|P\|_2^2\triangleq\sum_j|[z^j]P|^2\), the coefficient-norm form of
\eqref{eq:W-integral} gives
\[
 W_{3r+1,3}(2r)=\|(1+2z)C_r\|_2^2,\qquad
 W_{3r+1,3}(2r+1)=\|(1-z)C_r\|_2^2.
\]
Since
\[
 |1+2z|^2-|1-z|^2=3(1+z+z^{-1})
\]
on the unit circle, their difference divided by \(3\) is
\[
 \CT_z(1+z+z^{-1})C_r(z)C_r(z^{-1}).
\]
Now
\[
 C_r(z)C_r(z^{-1})
 =-z^{-(3r-1)}(z-1)^{4r-2}(2z^2+5z+2)^r.
\]
Using
\[
 1+z+z^{-1}=z^{-1}\frac{z^3-1}{z-1},
\]
the constant term is the negative of the difference between the
coefficients at degrees \(3r-3\) and \(3r\) in
\[
 H_r(z)\triangleq(z-1)^{4r-3}(2z^2+5z+2)^r.
\]
The polynomial \(H_r\) is anti-reciprocal of degree \(6r-3\), so these two
coefficients are \(-\gamma_r\) and \(\gamma_r\), respectively.  This proves
\eqref{eq:ternary-coefficient}.
\end{proof}

To prove positivity, we next convert \(\gamma_r\) into a finite sum and
derive a telescoping recurrence.  Substitution of
\(y=\sin^2(\theta/2)\) in
\eqref{eq:W-first-difference}, followed by
\eqref{eq:ternary-coefficient}, gives
\[
 \gamma_r=-\frac1{2\pi}\int_0^1
 \frac{(9-8y)^r(4y)^{2r-1}(4y-3)}
      {\sqrt{y(1-y)}}\,dy.
\]
After \(y=1-x\), expand \((1+8x)^r\).  Write \(B(a,b)\) for the beta
function and \((a)_j\triangleq a(a+1)\cdots(a+j-1)\), with \((a)_0\triangleq1\).  The
termwise integral is
\[
 \int_0^1x^{j-1/2}(1-x)^{2r-3/2}(1-4x)\,dx
 =B\left(j+\frac12,2r-\frac12\right)
  \frac{2r-3j-2}{2r+j}.
\]
Using
\[
 B\left(j+\frac12,2r-\frac12\right)
 =\frac{\pi}{4^{2r-1}}\binom{4r-2}{2r-1}
   \frac{(1/2)_j}{(2r)_j},
\]
we obtain
\begin{equation}
\label{eq:ternary-finite-sum}
 \gamma_r=\frac12\binom{4r-2}{2r-1}
 \sum_{j=0}^r\binom rj8^j\frac{(1/2)_j}{(2r)_j}
 \frac{3j-2r+2}{2r+j}.
\end{equation}

Define the coefficient polynomials
\begin{align*}
 u(r)&\triangleq486(r+1)(4r-1)(4r+1)(33r^2-4r-60),\\
 v(r)&\triangleq2(r+2)(3r+4)(3r+5)(33r^2-70r-23),\\
 c(r)&\triangleq104440+409424r+95070r^2-55242r^3+52272r^4.
\end{align*}
For \(r\geq3\), all three quantities are positive.  For \(c(r)\), this is
seen by writing \(x=r-3\):
\[
 c(r)=4930840+5133686x+2420580x^2
       +572022x^3+52272x^4.
\]

\begin{lemma}[Telescoping recurrence]
\label{lem:ternary-telescoper}
For every \(r\geq1\),
\begin{equation}
\label{eq:ternary-positive-recurrence}
 16v(r)(\gamma_{r+2}-16\gamma_{r+1})
 =u(r)(\gamma_{r+1}-16\gamma_r)+c(r)\gamma_{r+1}.
\end{equation}
\end{lemma}

\begin{proof}
Define
\[
 b_{s,j}\triangleq
 \frac12\binom{4s-2}{2s-1}\binom sj8^j
 \frac{(1/2)_j}{(2s)_j},\qquad
 a_{s,j}\triangleq b_{s,j}\frac{3j-2s+2}{2s+j},
\]
and set both quantities to zero outside \(0\leq j\leq s\).  Thus
\eqref{eq:ternary-finite-sum} says
\[
 \gamma_s=\sum_{j=0}^sa_{s,j}.
\]
The recurrence comes from a creative-telescoping certificate: with the
polynomial data below, \(P_0(r)a_{r,j}+P_1(r)a_{r+1,j}+P_2(r)a_{r+2,j}\) is a
difference of consecutive terms in \(j\), so summing over \(j\) gives the
recurrence.  Let
\[
 P_0(r)\triangleq u(r),\qquad P_2(r)\triangleq v(r),\qquad
 P_1(r)\triangleq-16v(r)-\frac{u(r)+c(r)}{16},
\]
and put
\[
 \delta_r\triangleq16r^3+64r^2+79r+30.
\]
Define
\begin{align*}
m_0(r)\triangleq{}&{-}9240-107456r-85526r^2+188518r^3\\
&\quad+231594r^4+45738r^5-13068r^6,\\
m_1(r)\triangleq{}&{-}48900-67930r+382778r^2+823841r^3\\
&\quad+601887r^4+211266r^5+39204r^6,\\
m_2(r)\triangleq{}&{-}200880-384462r+17313r^2+314580r^3\\
&\quad+143550r^4+13068r^5,\\
m_3(r)\triangleq{}&{-}9180+144648r+145953r^2-71145r^3-72171r^4,\\
m_4(r)\triangleq{}&59940+91476r-27135r^2-48114r^3,\\
m_5(r)\triangleq{}&14580+972r-8019r^2.
\end{align*}
Finally, put
\[
 H(r,j)\triangleq j\sum_{\ell=0}^5m_\ell(r)j^\ell,
 \qquad
 \mathcal R(r,j)\triangleq\frac{H(r,j)}{24\delta_r}.
\]
Exact simplification gives, for \(r\geq1\) and \(0\leq j\leq r+2\),
\begin{equation}
\label{eq:ternary-telescoper}
\begin{aligned}
 &P_0(r)a_{r,j}+P_1(r)a_{r+1,j}+P_2(r)a_{r+2,j}\\
 &\qquad=b_{r+2,j+1}\mathcal R(r,j+1)
             -b_{r+2,j}\mathcal R(r,j).
\end{aligned}
\end{equation}
For a direct check, divide by \(b_{r+2,j}\), clear denominators, and use
\[
 \frac{b_{s,j+1}}{b_{s,j}}
 =\frac{4(s-j)(2j+1)}{(j+1)(2s+j)},\qquad
 \frac{b_{s-1,j}}{b_{s,j}}
 =\frac{(s-j)(2s+j-2)(2s+j-1)}
        {4s(4s-5)(4s-3)}.
\]
After substitution, the difference between the two sides of
\eqref{eq:ternary-telescoper} is the zero polynomial in \(r,j\).
At the boundaries, \(\mathcal R(r,0)=0\) and \(b_{r+2,r+3}=0\).
Summing over \(j\) proves
\[
 P_0(r)\gamma_r+P_1(r)\gamma_{r+1}
 +P_2(r)\gamma_{r+2}=0.
\]
Multiplication by \(16\) and the definition of \(P_1\) give
\eqref{eq:ternary-positive-recurrence}.
\end{proof}

\begin{lemma}[Positivity of the ternary slope coefficient]
\label{lem:ternary-recurrence}
The integers \(\gamma_r\) are positive for every \(r\geq1\).
\end{lemma}

\begin{proof}
Direct coefficient extraction gives
\[
 \gamma_1=2,\quad \gamma_2=15,\quad \gamma_3=188,\quad
 \gamma_4=2977,\quad \gamma_5=54468,
\]
and \(\gamma_5-16\gamma_4=6836>0\).
Since \(u(r),v(r),c(r)>0\) for \(r\geq3\),
\eqref{eq:ternary-positive-recurrence}, applied successively for
\(r=4,5,\ldots\), gives
\[
 \gamma_{r+1}>16\gamma_r>0\qquad(r\geq4).
\]
Together with the initial values, this proves the lemma.
\end{proof}

The return probabilities used in
Proposition~\ref{prop:complete-monotone} satisfy
\[
 p_1=p_2=\frac13,\qquad p_3=\frac7{27},\qquad
 \Delta^2p_1=-\frac2{27}.
\]
Since \(\Delta f_{n,3}(w)=-3^{n-1}p_w\) by Proposition~\ref{prop:lambda}, as
in the proof of Proposition~\ref{prop:complete-monotone}, we have
\[
 (-1)^3\Delta^3f_{n,3}(1)=3^{n-1}\Delta^2p_1<0\qquad(n\geq4),
\]
so \(f_{n,3}\) is not completely monotonic.  This is the failure invoked in the
ternary case in Subsection~\ref{sec:small-alphabets}.

\section{Distance-kernel and binary specializations}
\label{app:binary-formulas}

The full distance-kernel formula follows from the lazy-walk representation
used in Proposition~\ref{prop:lambda}.  For \(1\leq w\leq n\),
\begin{equation}
\label{eq:lambda-multinomial}
 \lambda(w)=\frac{q^{n-w}}2
 \sum_{\substack{a,b\geq0\\a+b\leq w}}
 \binom{w}{a,b,w-a-b}(q-2)^{w-a-b}|a-b|.
\end{equation}
Indeed, after translating one of the two words to zero, each of the \(w\) differing
coordinates contributes \(1,-1,0\) to the distance difference with
probabilities \(1/q,1/q,(q-2)/q\), respectively.  Thus
\(\lambda(w)=q^n\E|S_w|/2\), and counting the three coordinate types gives
\eqref{eq:lambda-multinomial}.

For \(q=2\), only the terms with \(a+b=w\) remain, and hence
\[
 \lambda(w)
 =2^{n-w-1}\sum_{a=0}^w\binom wa|2a-w|
 =2^{n-w}w\binom{w-1}{\lceil w/2\rceil-1}.
\]
Equivalently,
\[
 \lambda(2i-1)=\lambda(2i)
 =2^{n-2i}i\binom{2i}{i},
\]
whenever \(2i\leq n\).  These agree with
\cite[Lemma~3.1, Eqs.~(17)--(18)]{Barg2021}.  Substitution in
\eqref{eq:Lambda} gives
\[
 \Lambda_{n,2}
 =\frac{n}{2^{n+1}}\binom{2n}{n}.
\]
Together with Theorem~\ref{thm:invariance}, this recovers
\cite[Prop.~3.2 and Thm.~3.3]{Barg2021}.

The closed binary evaluation of the squared Krawtchouk sum is
\begin{equation}
\label{eq:W-binary-closed}
 W_{n,2}(w)
 =\frac{\binom{2n-2w}{n-w}\binom{2w-2}{w-1}}
        {\binom{n-1}{w-1}},\qquad 1\leq w\leq n.
\end{equation}
By \eqref{eq:W-lambda-coefficient}, this is \cite[Eq.~(56)]{Barg2021}, which
Barg derives from the identity \cite[Eq.~(37)]{Barg2021}, attributed there to
\cite[Thm.~3.1.3]{FeinsilverFitzgerald1996}.  At \(n=2^m-1\), substitution of
\eqref{eq:W-binary-closed} into Theorem~\ref{thm:all-one-perfect} gives
\eqref{eq:binary-main-bound}.

\section{The ternary Golay certificate}
\label{app:golay-certificate}

The exact certificate for Theorem~\ref{thm:ternary-golay} gives the following
stronger inequality.

\begin{proposition}[Exact quadratic certificate]
\label{prop:golay-certificate}
Every code \(\C\subseteq\F_3^{11}\) of size \(3^6\) satisfies
\begin{equation}
\label{eq:golay-bound}
 \Dtwo(\C)\geq
 \frac{7\,446\,692+275\,076A_1(\C)+295\,002A_2(\C)}{3^{11}}
 \geq\frac{7\,446\,692}{3^{11}}.
\end{equation}
Equality in the final bound holds if and only if \(\C\) is a perfect
two-error-correcting code.
\end{proposition}

\begin{proof}
The ball \(B(0,2)\) has volume
\[
 V_2=1+2\binom{11}{1}+4\binom{11}{2}=243=3^5,
\]
and the two-error Lloyd polynomial is
\[
 c_w(2)=K_2^{(10,3)}(w-1)=\frac92(w-6)(w-9).
\]
By Proposition~\ref{prop:lloyd-support}, the dual support of a perfect code
lies in \(\{6,9\}\).  Equations \eqref{eq:mass} and
\eqref{eq:first-moment} force
\[
 \Adual_6+\Adual_9=242,\qquad
 6\Adual_6+9\Adual_9=1782,
\]
and hence
\begin{equation}
\label{eq:golay-spectrum}
 \Adual_6=132,\qquad \Adual_9=110.
\end{equation}

For an arbitrary \(3^6\)-point code \(\C\), the first two dual moments control
quadratic polynomials.  Define
\begin{equation}
\label{eq:golay-minorant}
 p(w)\triangleq323\,962-76\,236w+5\,463w(w-1).
\end{equation}
This is the unique quadratic interpolating \(W_{11,3}\) at
\(w=6,7,9\).  Direct evaluation of \eqref{eq:W-def} in exact integer
arithmetic, reproduced by the verification companion
\cite{Zabokritskiy2026Verifier}, gives the exact gaps
\[
\begin{array}{c|rrrrrr}
w&1&2&3&4&5&6\\ \hline
W_{11,3}(w)-p(w)
&653\,560\,947&17\,161\,110&1\,418\,310&176\,850&19\,836&0
\end{array}
\]
and
\[
\begin{array}{c|rrrrr}
w&7&8&9&10&11\\ \hline
W_{11,3}(w)-p(w)&0&18&0&14\,796&98\,460.
\end{array}
\]
Hence \(p(w)\leq W_{11,3}(w)\) for \(1\leq w\leq11\).

Lemma~\ref{lem:moments}, specialized to \(n=11,q=3,N=3^6\), gives
\begin{align*}
 \sum_{w=1}^{11}\Adual_w&=242,\\
 \sum_{w=1}^{11}w\Adual_w&=1782-81A_1,\\
 \sum_{w=1}^{11}w(w-1)\Adual_w&=11880-1080A_1+54A_2.
\end{align*}
Consequently,
\[
 \sum_{w=1}^{11}\Adual_wW_{11,3}(w)
 \geq7\,446\,692+275\,076A_1+295\,002A_2.
\]
Theorem~\ref{thm:spectral} proves the first inequality in
\eqref{eq:golay-bound}.  For a perfect code, \(A_1=A_2=0\), and
\eqref{eq:golay-spectrum} together with
\[
 W_{11,3}(6)=30\,436,\qquad W_{11,3}(9)=31\,174
\]
shows that the bound is attained.

Conversely, equality forces \(A_1=A_2=0\), while the gap table restricts the
dual support to \(\{6,7,9\}\).  Write the corresponding masses as
\(b_6,b_7,b_9\).  The three moment identities give
\[
\begin{aligned}
 b_6+b_7+b_9&=242,\\
 6b_6+7b_7+9b_9&=1782,\\
 30b_6+42b_7+72b_9&=11880.
\end{aligned}
\]
Their unique solution is
\[
 (b_6,b_7,b_9)=(132,0,110).
\]
Thus the dual support of \(\C\) lies in the zero set \(\{6,9\}\) of
\(c_w(2)\).  Since \(3^6V_2=3^{11}\), the converse in
Proposition~\ref{prop:lloyd-support} shows that \(\C\) is a perfect
two-error-correcting code.
\end{proof}

\section{The binary Golay certificate}
\label{app:binary-golay-certificate}

This appendix proves Theorem~\ref{thm:binary-golay} in exact integer
arithmetic.  For \(j\geq0\), write
\[
 x^{\underline j}\triangleq x(x-1)\cdots(x-j+1),
 \qquad x^{\underline0}\triangleq1.
\]
The required higher moments follow from the same coordinate character sum as
Lemma~\ref{lem:moments}.  They are the Delsarte form, valid for arbitrary
codes, of the standard Pless power-moment identities
\cite{Pless1963,Delsarte1973,MacWilliamsSloane1977}, stated in the
normalization of the certificate.

\begin{lemma}[Factorial dual moments]
\label{lem:factorial-moments}
For every \(N\)-point code \(\C\subseteq G^n\) and \(1\leq j\leq n\),
\begin{equation}
\label{eq:factorial-moments}
 \sum_{w=1}^n w^{\underline j}\Adual_w
 =\frac{q^{n-j}}N
 \sum_{i=0}^j(-1)^i
 j^{\underline i}(n-i)^{\underline{j-i}}
 (q-1)^{j-i}A_i.
\end{equation}
\end{lemma}

\begin{proof}
The factor \(\wt(\xi)^{\underline j}\) counts ordered \(j\)-tuples of
distinct coordinates on which \(\xi\) is nontrivial.  Fix a difference
vector of weight \(i\).  The chosen coordinates must contain its support;
otherwise the character sum vanishes.  If \(i\leq j\), the number of ordered
choices containing that support is
\[
 j^{\underline i}(n-i)^{\underline{j-i}}.
\]
Character orthogonality contributes \(-1\) on each of its \(i\) nonzero
coordinates, \(q-1\) on each remaining chosen coordinate, and \(q\) on every
unchosen coordinate.
Summing over ordered codeword pairs at distance \(i\) gives
\eqref{eq:factorial-moments}.
\end{proof}

At the binary Golay parameters,
\[
 n=23,\qquad q=2,\qquad N=2^{12},
\]
put
\[
 M_j\triangleq\sum_{w=1}^{23}w^{\underline j}\Adual_w.
\]
The mass identity and Lemma~\ref{lem:factorial-moments} give
\begin{align}
 M_0&=2047,\label{eq:binary-golay-m0}\\
 M_1&=23\,552-1\,024A_1,\label{eq:binary-golay-m1}\\
 M_2&=259\,072-22\,528A_1+1\,024A_2,\label{eq:binary-golay-m2}\\
 M_3&=2\,720\,256-354\,816A_1+32\,256A_2-1\,536A_3,
 \label{eq:binary-golay-m3}\\*
 M_4&=27\,202\,560-4\,730\,880A_1+645\,120A_2
       -61\,440A_3+3\,072A_4.
 \label{eq:binary-golay-m4}
\end{align}

The radius-three ball has volume
\[
 V_3=1+\binom{23}{1}+\binom{23}{2}+\binom{23}{3}
 =2048=2^{11},
\]
and the three-error Lloyd polynomial factors as
\begin{equation}
\label{eq:binary-golay-lloyd}
 c_w(3)=K_3^{(22,2)}(w-1)
 =-\frac43(w-8)(w-12)(w-16).
\end{equation}
Thus the Lloyd roots are \(8,12,16\).

Define
\begin{equation}
\label{eq:binary-golay-minorant}
 p(w)\triangleq
 705\,432-8\,992(w-12)^2+10\,000(w-12)^4.
\end{equation}
In the falling-factorial basis,
\begin{equation}
\label{eq:binary-golay-falling}
\begin{aligned}
 p(w)={}&206\,770\,584-60\,743\,184w
       +7\,261\,008w^{\underline2}\\
       &-420\,000w^{\underline3}
       +10\,000w^{\underline4}.
\end{aligned}
\end{equation}
The closed binary formula \eqref{eq:W-binary-closed} and the symmetry
\eqref{eq:binary-symmetry} give the following complete gap certificate,
computed in exact integer arithmetic and reproduced by the verification
companion \cite{Zabokritskiy2026Verifier}.
Only \(1\leq w\leq12\) need be displayed, because both \(W_{23,2}(w)\)
and \(p(w)\) are invariant under \(w\mapsto24-w\).  Put
\[
\Delta(w)\triangleq W_{23,2}(w)-p(w).
\]
\[
\begin{array}{c|r@{\qquad}c|r}
w&\Delta(w)&w&\Delta(w)\\ \hline
1&2\,103\,952\,936\,320&2&48\,832\,727\,808\\
3&3\,514\,842\,240&4&417\,939\,456\\
5&62\,568\,576&6&8\,989\,440\\
7&713\,088&8&0\\
9&180\,096&10&187\,136\\
11&66\,176&12&0
\end{array}
\]
Consequently,
\[
 p(w)\leq W_{23,2}(w)\qquad(1\leq w\leq23),
\]
with equality exactly at \(w=8,12,16\).

Substituting \eqref{eq:binary-golay-m0}--\eqref{eq:binary-golay-m4}
into \eqref{eq:binary-golay-falling} gives
\begin{equation}
\label{eq:binary-golay-moment-evaluation}
\begin{aligned}
 \sum_{w=1}^{23}\Adual_wp(w)
 ={}&3\,277\,860\,456
  +338\,952\,192(A_1+A_2)\\
 &+30\,720\,000(A_3+A_4).
\end{aligned}
\end{equation}
The spectral formula therefore yields the stronger bound
\begin{equation}
\label{eq:binary-golay-certificate-bound}
\begin{aligned}
 \Dtwo(\C)
 &\geq
 \frac{
 3\,277\,860\,456
 +338\,952\,192(A_1+A_2)
 +30\,720\,000(A_3+A_4)}
 {2^{23}}\\
 &\geq
 \frac{3\,277\,860\,456}{2^{23}}
 =\frac{409\,732\,557}{1\,048\,576}.
\end{aligned}
\end{equation}

If \(\C\) is a perfect three-error-correcting code, then its minimum distance is
at least seven, so \(A_1=\cdots=A_4=0\), and
Proposition~\ref{prop:lloyd-support} restricts its dual support to
\(\{8,12,16\}\).  Since \(p=W_{23,2}\) there, every inequality in
\eqref{eq:binary-golay-certificate-bound} is an equality.

Conversely, equality in the final bound forces \(A_1=\cdots=A_4=0\), and
the strict gaps above restrict the dual support to the zeros
\(8,12,16\) of \(c_w(3)\) in \eqref{eq:binary-golay-lloyd}.  Since
\(NV_3=2^{12}2^{11}=2^{23}\), the converse in
Proposition~\ref{prop:lloyd-support} shows that \(\C\) is a perfect
three-error-correcting code.

\begin{remark}[The cubic obstruction]
\label{rem:binary-golay-cubic}
The quadratic interpolant of \(W_{23,2}\) at the Lloyd roots is
\[
 I_2(w)=705\,432+151\,008(w-12)^2.
\]
Every cubic through the same three points has the form
\[
 I_2(w)+\tau(w-8)(w-12)(w-16).
\]
Minorization at \(w=9\) requires \(\tau\leq-21\,424\), whereas
minorization at \(w=15\) requires \(\tau\geq21\,424\).  Thus no cubic can
play the role of \(p\).
\end{remark}

\section*{Computational verification}
A verification companion archived on Zenodo \cite{Zabokritskiy2026Verifier}
reproduces, in exact integer, rational, and symbolic arithmetic, the
telescoping certificate of Lemma~\ref{lem:ternary-telescoper} and the initial
values in Lemma~\ref{lem:ternary-recurrence}, the quadrature identity of
Lemma~\ref{lem:two-node-quadrature}, the gap tables in
Appendices~\ref{app:golay-certificate} and~\ref{app:binary-golay-certificate},
and the stated polynomial and moment identities.  The proofs state all
certificate data.